\documentclass[11pt]{article}

\usepackage{fullpage}
\usepackage{hdb_macros}
\usepackage{algorithm}
\usepackage{algpseudocode}
\usepackage{mdframed}
\usepackage{framed}

\usepackage{tikz}
\usepackage{csquotes}

\usepackage[multiuser,inline,nomargin]{fixme}
\fxusetheme{color}
\FXRegisterAuthor{g}{eh}{\color{blue}Surendra}
\FXRegisterAuthor{k}{ek}{\color{purple}Karthik}
\FXRegisterAuthor{z}{ee}{\color{cyan}Zeyong}
\FXRegisterAuthor{s}{es}{\color{red}Sidhant}

\newcommand{\prob}[1]{\textsc{#1}}

\title{Downward self-reducibility in the total function polynomial hierarchy}

\author{
    Karthik Gajulapalli\thanks{Georgetown University. Email: \email{kg816@georgetown.edu}. Supported by NSF grant CCF-2338730.}
    \and
    Surendra Ghentiyala\thanks{Cornell University. Email: \email{sg974@cornell.edu}. This work is supported in part by the NSF under Grants Nos.~CCF-2122230 and CCF-2312296, a Packard Foundation Fellowship, and a generous gift from Google.} 
    \and
    Zeyong Li\thanks{National University of Singapore. Email: \email{zeyong@u.nus.edu}. Supported by NRF grant NRF-NRFI09-0005.}
    \and
    Sidhant Saraogi\thanks{Georgetown University. Email: \email{ss456@georgetown.edu}.}
}

\date{\today}

\begin{document}
\pagenumbering{roman}
\maketitle

\begin{abstract}
    A problem $\mathcal{P}$ is considered $\emph{downward self-reducible}$, if there exists an efficient algorithm for $\mathcal{P}$ that is allowed to make queries to only strictly smaller instances of $\mathcal{P}$. Downward self-reducibility has been well studied in the case of decision problems, and it is well known that any downward self-reducible problem must lie in $\PSPACE$. Harsha, Mitropolsky and Rosen [ITCS, 2023] initiated the study of downward self reductions in the case of search problems. They showed the following interesting collapse: if a problem is in \cc{TFNP} and also downward self-reducible, then it must be in \cc{PLS}. Moreover, if the problem admits a unique solution then it must be in \cc{UEOPL}.

    We demonstrate that this represents just the tip of a much more general phenomenon, which holds for even harder search problems that lie higher up in the total function polynomial hierarchy ($\cc{TF\Sigma_i^P}$). In fact, even if we allow our downward self-reduction to be much more powerful, such a collapse will still occur. 
    We show that any problem in $\cc{TF\Sigma_i^P}$ which admits a randomized downward self-reduction  with access to a $\cc{\Sigma_{i-1}^P}$ oracle must be in $\cc{PLS}^{\cc{\Sigma_{i-1}^P}}$. If the problem has \textit{essentially unique solutions} then it lies in  $\cc{UEOPL}^{\cc{\Sigma_{i-1}^P}}$. 

    As an application of our framework, we get new upper bounds for the problems \prob{Range Avoidance} and \prob{Linear Ordering Principle} and show that they are both in $\cc{UEOPL}^{\NP}$, a particularly small subclass of $\cc{TF\Sigma_2^P}$. As a corollary of the powerful Range Avoidance framework, we get that a host of explicit construction problems like constructing rigid matrices, Ramsey graphs, hard truth tables against fixed polynomial size circuits are all in $\cc{UEOPL}^{\NP}$. This appears to be an orthogonal containment to the results  by Chen, Hirahara and Ren [STOC 2024],  Li [STOC 2024], and Korten and Pitassi [FOCS 2024] that put the \prob{Linear Ordering Principle} and hence \prob{Range Avoidance} in $\cc{FS_2P}$.

    In the third level of the polynomial hierarchy, we show that \prob{King}, the only candidate problem not known to collapse to any smaller sub-class of $\cc{TF\Sigma_3^P}$, indeed collapses to $\cc{PLS}^{\cc{\Sigma_2^P}}$. Even more surprisingly, we give a $\ZPP^{\cc{\Sigma_2^P}}$ algorithm for \prob{King}. This refutes the idea proposed in Kleinberg, Korten, Mitropolosky, and Papadimitriou [ITCS 2021] that \prob{King} is in a class in and by itself.

    Along the way, we highlight the power of our framework by giving alternate proofs of  \cc{PLS} and \cc{UEOPL} membership for several important problems: the P-matrix linear complementarity problem, finding a winning strategy in a parity game, finding a Tarski fixed point, and the \prob{S-Arrival} problem. These proofs only rely on the fact that the natural recursive algorithms for these problems yield downward self-reductions, and therefore classifications into \cc{PLS}/ \cc{UEOPL}. Consequently, our proofs are shorter, simpler, and take a more algorithmic perspective.
    

    Finally, we highlight the limitations of downward self-reducibility as a framework for \cc{PLS} membership. We show that unless \cc{FP}=\cc{PLS}, there exists a \cc{PLS}-complete problem that is not d.s.r. This indicates that traditional d.s.r is not a complete framework for \cc{PLS}, answering a question left open by Harsha, Mitropolsky and Rosen [ITCS, 2023].
        
\end{abstract}

\thispagestyle{empty}
\newpage

\tableofcontents
\newpage
\pagenumbering{arabic}

\section{Introduction}

A problem is considered \emph{downward self-reducible}, if given solutions to smaller sub-problems, one can solve the original problem efficiently. More precisely, a search problem $\mathcal{P}$ is traditionally \emph{downward self-reducible} (d.s.r)\footnote{We often use d.s.r as an abbreviation for both downward self-reduction and downward self-reducible. It should be clear from context which one we mean.} if given an instance $x \in \mathcal{P}$, where $|x| = n$, there is a polynomial time reduction to a collection of problems $x_1, \cdots x_t \in \mathcal{P}$ such that for each $x_i$, we have that $|x_i| < n$, and there is a polynomial time procedure $\mathcal{A}$ that given solutions to all $x_i$ outputs a solution for $x$.

It turns out that many natural problems are downward self-reducible. In fact downward self-reducibility has found itself as one of the central tools in establishing search-to-decision reductions for a plethora of problems. Perhaps the most well-known example of such a problem is \prob{FSAT} (functional version of \prob{SAT}), also known as \prob{Search-SAT}.
\begin{maindefinition} The problem \prob{FSAT} takes as input a boolean formula $\varphi$ over $n$ variables $x_1, \dots, x_n$, and $t= \poly(n)$ clauses. If $\varphi$ is satisfiable output any satisfying assignment, otherwise, output ``UNSAT''
\end{maindefinition}

The downward self-reduction for \prob{FSAT} proceeds as follows. Let $\varphi_0$ denote the sub-formula of $\varphi$ where $x_1$ is set to $0$ and then all clauses are maximally simplified\footnote{For example, $(x_1 \land x_2) \lor (x_3)$ gets simplified to  $x_3$ since $(x_1 \land x_2)$ is false under the assumption that $x_1 = 0$.}. Similarly, let $\varphi_1$ denote $\varphi$ where $x_1$ is set to $1$ and then all clauses are maximally simplified. Call the \prob{FSAT} oracle on $\varphi_0$ and $\varphi_1$. If either the oracle call on $\varphi_0$ or $\varphi_1$ returns a satisfying assignment $x_2 = \alpha_2, \dots, x_n = \alpha_n$, output $(0, \alpha_2, \dots, \alpha_n)$ or $(1, \alpha_2, \dots, \alpha_n)$ respectively. Otherwise, output ``UNSAT''. Notice that this reduction is correct since $\varphi$ is satisfiable if and only if either $\varphi_0$ or $\varphi_1$ is satisfiable. Furthermore, $|\varphi_0| < t$ and $|\varphi_1| < t$, so the reduction is also downward.

Downward self-reductions similar to the one for \prob{FSAT} can be shown for a wide variety of different \NP-complete problems such as \prob{VertexCover} and \prob{Hamiltonian Path}. In fact, it is known that any problem that is downward self-reducible must be contained in $\PSPACE$. This follows from the existence of the  $\PSPACE$ algorithm that simply runs the recursive algorithm defined by the downward self-reduction. Moreover, this bound is known to be tight since there exist \PSPACE-complete problems which are downward self-reducible (e.g., the TQBF problem).

The recent work of \cite{harsha2022downward} asks about the complexity of a special type of downward self-reducible search problems: \cc{TFNP} problems. Informally, a search problem is in \cc{TFNP} (total function nondeterministic polynomial time) if every instance always has a solution (total function) and given a solution it is efficiently verifiable in polynomial time (nondeterministic polynomial time). Restricting \cc{TFNP} even further, one can define the class \cc{TFUP}, that include all problems in \cc{TFNP} that have a unique solution.
Perhaps the most useful example to have in mind is the problem \prob{Factor}: given an $n$ bit integer $x$, output $0^{n-1}$ if $x$ is prime, otherwise output an $n-1$ bit integer $y$ such that $y$ divides $x$. Notice that \prob{Factor} is a total function since every integer $x$ is either prime or has a non-trivial factor of size at most $x/2$. Furthermore, the solutions to \prob{Factor} are efficiently verifiable since checking if $x$ is prime or $y$ divides $x$ can be done in polynomial time. While \prob{Factor} is not in \cc{TFUP} since an integer may have many factors, it is possible to define a search problem that outputs the full factorization of an $n$ bit input $x$, which is in \cc{TFUP} by the unique factorization theorem.

\cite{harsha2022downward} show that any problem that is in \cc{TFNP} and also downward self-reducible must lie in \cc{PLS} (Polynomial Local Search), a well-studied \cc{TFNP} subclass. Moreover, if the problem is in \cc{TFUP} and downward self-reducible then it lies in \cc{UEOPL} (Unique End Of Potential Line), a particularly small \cc{TFNP} subclass.
These results should be somewhat surprising as downward self-reducibility of a problem $\mathcal{P}$ generally only implies that $\mathcal{P} \in \cc{PSPACE}$, but adding the minor additional constraint that $\mathcal{P}$ is in \cc{TFNP}/\cc{TFUP} (i.e. $\mathcal{P}$  is total and efficiently verifiable) collapses the complexity of $\mathcal{P}$ to relatively small \cc{TFNP} subclasses \cc{PLS}/\cc{UEOPL}. As a consequence \cite{harsha2022downward} gives a barrier for designing recursive algorithms for \prob{Factor}. A recursive algorithm for \prob{Factor} would factor an $n$ bit integer by factoring several integers of bit-length at most $n-1$, and thus imply a downward self-reduction for \prob{Factor}. This would place \prob{Factor} in $\cc{UEOPL}$, resolving long-standing open questions about the complexity of \prob{Factor} within \cc{TFNP}.

Based on the collapses in \cc{TFNP} for problems that admit downward self-reductions, it is natural to ask if this is representative of a more general phenomenon. We explore exactly this question in search problems which admit downward self-reductions and lie in \cc{TFPH} (Total Function Polynomial Hierarchy). In what has now become a seminal paper Kleinberg, Korten, Mitropolsky and Papadimitriou \cite{kleinberg2021total} introduce $\cc{TFPH} = \bigcup_{i=1}^{\infty} \cc{TF\Sigma_i^P}$, a version of the polynomial hierarchy defined on total search problems. The class $\cc{TF\Sigma_i^P}$ defines a total search problem, where given a solution it is verifiable by a $\cc{\Pi_{i-1}^P}$ machine ($\cc{TFU\Sigma_i^P}$ when the problem admits a unique solution). Note that when $i=1$, this corresponds to \cc{TFNP} (and resp. \cc{TFUP}). 

In the second level of the polynomial hierarhcy $\cc{TF\Sigma_2^P}$, \cite{kleinberg2021total} introduce the problem \prob{Range Avoidance} (\prob{Avoid}) whose study has led to breakthrough results in proving maximum circuit lower bounds, a major open problem in complexity theory \cite{chen2024symmetric, li2024symmetric}. More generally, \cite{korten2022hardest} showed that algorithms for \prob{Avoid} would have many interesting consequences for both derandomization and getting explicit constructions of many combinatorial objects. 
As one application, \cite{korten2022hardest}  showed an equivalence between the existence of $\cc{FP}^{\NP}$ algorithms for \prob{Avoid} and the existence of a language in $\E^{\NP}$ with circuit complexity $2^{\Omega(n)}$. To put this in perspective, the best lower bound currently known for $\E^{\NP}$ is that it cannot be computed by circuits of size $\sim 3.1n$ \cite{li20221}. However, the best algorithm for \prob{Avoid} lies in the class $\cc{FS_2P}$ \cite{chen2024symmetric,li2024symmetric, korten2024LOP} which leaves the following question about the complexity of \prob{Avoid} open.

\begin{mainopenproblem}\label{op:fp_np_avoid}
    Is there a $\cc{FP}^{\NP}$ algorithm for \prob{Avoid}?
\end{mainopenproblem}

Going even higher up to the third level of the polynomial hierarchy, \cite{kleinberg2021total} introduce two problems in $\cc{TF\Sigma_3^P}$: \prob{Shattering} (the search problem derived from the Sauer-Shelah lemma on the shattering dimension of sets), and \prob{King} (the principle that any tournament has a sort of almost-maximal participant). In the case of  \prob{Shattering}, they show that the problem collapses to the subclass $\cc{PPP}^{\cc{\Sigma_2^P}}$
\footnote{\cc{PPP} is another subclass of \cc{TFNP} corresponding to the application of the pigeonhole principle.}. However, in the case of \prob{King} they posit that unless $\#\P$ is in the polynomial hierarchy, it seems unlikely that \prob{King} belongs to $\cc{PLS}^{\cc{\Sigma_2^P}}$. This leads to the natural open question.

\begin{mainopenproblem}\label{op:king_pls_sig2}
    Is there a $\cc{PLS}^{\cc{\Sigma_2^P}}$ algorithm for \prob{King}?
\end{mainopenproblem}

\subsection{Our Results}
In this work, we show that \cc{PLS}/\cc{UEOPL} membership theorems of \cite{harsha2022downward} are actually just the tip of the iceberg of a much more general phenomenon. The downward self-reducibility framework for showing \cc{PLS} membership in fact works for all search problems with \emph{efficiently verifiable solutions} and \emph{promise preserving downward self-reductions}. Thus, we are able to fully recover the results of \cite{harsha2022downward} as the totality of any $\cc{TF\Sigma_i^P}$ problem ensures that all downward self-reductions are promise-preserving\footnote{A total problem satisfies the trivial promise, every input is in the relation.}, and all solutions are verifiable in \cc{P} for \cc{TFNP} problems. We are further able to generalize to randomized reductions, as well as problems which are not strictly downward self-reducible (\cref{def: mu-dsr}). Most interestingly, we show that the downward self-reducibility framework relativizes as we go up the polynomial hierarchy, i.e. a problem $\mathcal{P}$ that is in $\cc{TF\Sigma_i^P}$ and also downward self-reducible with access to a $\cc{\Sigma_{i-1}^P}$ oracle lies in $\cc{PLS^{\Sigma_{i-1}^P}}$.

We then use these ideas to establish surprising containments in \cc{TF\Sigma_i^P} subclasses for several total function problems: \prob{Avoid}, \prob{LinearOrderingPrinciple}, and \prob{King}. As a result we make progress on \cref{op:fp_np_avoid} and fully resolve \cref{op:king_pls_sig2}. Along the way, we highlight the power of our framework by giving simple alternate proofs of membership in \cc{PLS} and \cc{UEOPL} of some classic \cc{TFNP} problems.
 
\subsubsection{Downward Self Reductions in the Total Function Polynomial Hierarchy}

We first highlight our main technical contribution: a framework (\cref{sec:mu_dsr}) for showing collapses in $\cc{TF\Sigma_i^P}$ for problems that are downward self-reducible, under very powerful notions of downward self-reducibility. One can interpret our framework as a strong generalization of the main result from \cite{harsha2022downward}, which we recall below:

\begin{theorem*}[\cite{harsha2022downward}]\label{thm: harsha}
    Every (traditionally) downward self-reducible problem in \cc{TFNP} is in \cc{PLS}. Moreover, every downward self-reducible problem in \cc{TFUP} is in \cc{UEOPL}.
\end{theorem*}

While this result is already quite powerful, it is a somewhat restrictive framework. For example, it requires that the downward oracle calls are to instances of strictly smaller ``size''. The result holds only for problems in the first level of the polynomial hierarchy, and their framework requires the problems to be total. Moreover, the downward self-reduction must run in deterministic polynomial time. In contrast, we will allow for much more powerful downward self-reducibility and prove the following more general theorem (see \cref{thm:mu-dsr}).

\begin{maintheorem}
\label{thm:main_mu-dsr}
    Let $\mathcal{R}$ be a search problem in \cc{PromiseF\Sigma_i^P} which has a (randomized) $\mu$-d.s.r. If $\mu(x) \leq p(|x|)$ for a polynomial $p$, then there is a (randomized) reduction from $\mathcal{R}$ to $\PLS^{\cc{\Sigma_{i-1}^P}}$. Furthermore, if $\mathcal{R}$ has unique solutions, then there is a (randomized) reduction from $\mathcal{R}$ to $\cc{UEOPL}^{\cc{\Sigma_{i-1}^P}}$.
\end{maintheorem}

We now contrast \cref{thm:main_mu-dsr} to the original framework in \cite{harsha2022downward}, and describe our relaxations with some observations about our proof that let us generalize the downward self-reducible framework in the following ways. 

\begin{enumerate}
    \item Our theorem applies to \cc{PromiseF\Sigma_i^P} problems rather than just \cc{TFNP} problems. Informally, a $\cc{PromiseF\Sigma_i^P}$ is a promise search problem (one where only certain inputs are allowed) which has a $\cc{\Sigma_{i-1}^P}$ verifier (see \cref{def: promiseTFSigma_i}). We require that our downward self-reduction be promise-preserving\footnote{Our downward self-reduction only maps to instances that also satisfy the promise.}, and as a special case of \cref{thm:main_mu-dsr}, we get that a downward self-reducible  $\cc{TF\Sigma_i^P}$/$\cc{TFU\Sigma_i^P}$ problem is in $\PLS^{\cc{\Sigma_{i-1}^P}}$/$\cc{UEOPL}^{\cc{\Sigma_{i-1}^P}}$ (see \cref{cor: mu-d.s.r}). This generalization, allows us to provide a very simple proof that \prob{Promise-P-LCP} reduces to $\cc{UEOPL}$ (see \cref{cor:plcp_ueopl}).

    \item The traditional notion of downward self-reducibility encapsulates a set of recursive algorithms, where the recursive calls are made on strictly smaller instances. However, when designing recursive algorithms, the notion of ``smaller'' is often more flexible than just instance size. To capture this idea, we define $\mu$-d.s.r (see \cref{def: mu-dsr}). Informally, a problem $R$ is $\mu$-d.s.r if there exists a polynomially bounded function $\mu$ defined on instances of $R$ such that solving $R$ on an instance $I$ is polynomial-time reducible to solving $R$ on instance $I_1, \dots, I_{t}$ such that $\mu(I_i) < \mu(I)$. We believe that the notion of $\mu$-d.s.r is a more accurate abstraction of recursive algorithms since recursive algorithms have some measure on which the instance decreases at each recursive call, but not necessarily the size.

    One of the key takeaways from \cref{thm:main_mu-dsr} and this paper as a whole is that \emph{recursive algorithms and inductive proofs of totality are closely related to \cc{PLS}}, since these correspond to $\mu$-d.s.r, as we will see in \cref{sec: PLS_TFNP}. For a concrete example, consider the problem \prob{Tarski} for which we recover an alternate proof of $\cc{PLS}$ membership. In our reduction, it is crucial for us to have the notion of $\mu$-d.s.r over traditional d.s.r (see \cref{sec:tarski}).

    \item For problems in $\cc{PromiseF\Sigma_i^P}$, we will allow the downward self-reduction to make use of a $\cc{\Sigma_{i-1}^P}$ oracle, which makes the reduction very powerful. 

    \item We further generalize the fact that a $\mu$-d.s.r $\cc{TFU\Sigma_i^P}$ problem $\mathcal{R}$ with polynomially bounded $\mu$ reduces to $\cc{UEOPL}^{\cc{\Sigma_{i-1}^P}}$. We observe, that when $i \geq 2$, unique solutions are no longer required for this theorem to hold. It suffices that $\mathcal{R}$ has what are dubbed as essentially unique solution \cite{korten2024LOP} (see \cref{def: essentially_unique} and \cref{lem: essentially_unique_ueopl}). We use this property, crucially in showing that \prob{Linear Ordering Principle} and \prob{Avoid} are in $\cc{UEOPL}^{\NP}$ instead of just $\cc{PLS}^{\NP}$(see \cref{sec:lop}).

    \item Our results apply to randomized reductions.
    
\end{enumerate}

As a corollary of our framework we show that even though \prob{FSAT} has a downward self-reduction, \prob{Promise-FSAT} is very unlikely to have a promise-preserving downward self-reduction.

\begin{maincorollary}
    If \prob{Promise-FSAT} has a promise-preserving downward self-reduction then $\NP = \coNP$
\end{maincorollary}

This follows from just observing that such a reduction would reduce \prob{SAT} to \cc{PLS}, which would imply $\NP = \coNP$. A stronger collapse follows if we consider the problem \prob{PromiseF-UniqueSAT} (see \cref{cor: promise_fsat}).

\subsubsection{New Upper Bounds for \prob{Avoid} and \prob{Linear Ordering Principle}}

Equipped with \cref{thm:main_mu-dsr}, we make progress towards \cref{op:fp_np_avoid}. We consider two problems: \prob{Avoid} and \prob{Linear Ordering Principle} (\prob{LOP}), and show that they are both in $\cc{UEOPL}^{\NP}$. 

We begin by defining the range avoidance problem (\prob{Avoid}) introduced in \cite{kleinberg2021total}
\begin{maindefinition}[Range Avoidance (\prob{Avoid})]
    Given a polynomial sized circuit $C: \bit^n \rightarrow \bit^{n+1}$, find an element $y$ not in the range of $C$.
\end{maindefinition}

\prob{Avoid} is total via the dual pigeonhole principle, and a solution $y$ to \prob{Avoid} can be verified by the following $\coNP$ predicate: $\forall x:$  $C(x) \neq y$. This puts \prob{Avoid} in $\cc{TF\Sigma_2^P}$. Moreover, \prob{Avoid} has an abundance of solutions, since at least half the elements in the co-domain are not in the range of $C$. This abundance guarantees that \prob{Avoid} is in $\cc{TFZPP}^{\NP}$, via the trivial algorithm that guesses a random string and checks if it is a valid solution.

\prob{Avoid} has received a lot of attention in recent literature because of its connection to derandomization \cite{korten2022hardest, korten2022derandomization} and explicit constructions \cite{korten2022hardest, ren2022range, guruswami2022range, gajulapalli2023range, chen2023range, gajulapalli2024oblivious}. For example \cite{korten2022hardest} showed that if $\prob{Avoid} \in \mathfrak{C}$ for some complexity class $\mathfrak{C}$, then $\BPP \subseteq \mathfrak{C}$. Moreover, the construction of many combinatorial objects known to exist via the probabilistic method like hard truth tables, rigid matrices and Ramsey graphs, just to name a few, all reduce to $\prob{Avoid}$. However to get explicit constructions of these combinatorial objects, it is not sufficient to just have an algorithm to \prob{Avoid}, we also require that our algorithm isolates a single solution for a given instance (we call this problem single-value \prob{Avoid}). Consider for instance the task of constructing a language that cannot be computed by circuits of size $n^2$. If one tries to define a language based on the output of the hard truth table generated from an \prob{Avoid} algorithm that is not single-valued, then the language would not be well-defined since different hard truth tables will define different languages.

We now show the power of single-value \prob{Avoid} as a tool for capturing most explicit construction problems via the following folklore theorem:

\begin{theorem*}[\cite{korten2022hardest, guruswami2022range}]
    If Single-Value \prob{Avoid} $\in \mathfrak{C}$ for some complexity class $\mathfrak{C}$\footnote{We only consider classes $\mathfrak{C}$ where $\P \subseteq \mathfrak{C}$.}, then we can output the explicit construction of the following objects in $\mathfrak{C}$: Ramsey Graphs, Two Source Extractors, Rigid Matrices, Linear Codes, Hard Truth Tables for any Fixed Polynomial Sized Circuit, $\cc{K}^{\poly}$-random strings.
\end{theorem*}

We now state our main result of this section. 

\begin{maintheorem}\label{thm:main_avoid_ueopl}
\prob{Avoid} is in $\cc{UEOPL}^{\NP}$.    
\end{maintheorem}

In fact, our proof also guarantees uniqueness of the solution. Combined with the theorem on explicit constructions from single-value \prob{Avoid} we get the following corollary.


\begin{maincorollary}\label{cor:main_explicit_construction}
    There is an explicit construction of Ramsey Graphs, Two Source Extractors, Rigid Matrices, Linear Codes, Hard Truth Tables against any Fixed Polynomial Sized Circuit, $\cc{K}^{\poly}$-random strings in $\cc{UEOPL}^{\NP}$
\end{maincorollary}

We first give a proof overview of \cref{thm:main_avoid_ueopl} (see \cref{cor: LOP_in_UEOPL}), and then compare our results to other upper bounds on \prob{Avoid}.

It is not clear that \prob{Avoid} gives us any structure to directly construct a $\mu$-downward self-reduction. So to apply our framework, we consider an intermediate problem called \prob{Linear Ordering Principle} (\prob{LOP}) introduced in \cite{korten2024LOP}, where given as input a total order: $\prec\;$, one must find the smallest element in $\prec\;$. If it is a total order, \prob{LOP} has a unique solution, and it is not hard to see that this problem admits essentially unique solutions (\cref{def: essentially_unique}).
\begin{maindefinition}[Linear Ordering Principle (\prob{LOP}) \cite{korten2024LOP}]
    Given $\prec\;: \{0, 1\}^n \times \{0, 1\}^n \rightarrow \{0, 1\}$, where $\prec\;$ is described by a polynomial sized circuit, either:
    \begin{enumerate}
        \item Find a witness that $\prec\;$ does not define a total order i.e. $x, y, z \in \{0, 1\}^n$ such that one of the following holds: (i) $x \prec\; x$, (ii) $x \neq y$, $x \not \prec\; y$ and $y \not \prec\; x$ or (iii) $x \prec\; y \prec\; z$ and $x \not \prec\; z$. 
        \item Find the minimal element of the order defined by $\prec\;$. 
    \end{enumerate}
\end{maindefinition}

Moreover, given a solution $y$ to \prob{LOP} one can verify that $y$ is minimal by the following $\coNP$ predicate: $\forall x:$ $\prec\;(y, x) = 1$ which puts \prob{LOP} in $\cc{TF\Sigma_2^P}$.  We then show that \prob{LOP} is $\mu$-downward self-reducible (\cref{lop:dsr}), thus by \cref{thm:mu-dsr} we get the following theorem:

\begin{maintheorem} \label{thm:main_lop_ueopl_np}
    \prob{LOP} is in $\cc{UEOPL}^{\NP}$
\end{maintheorem}

Applying the polynomial time many-to-one reduction from \prob{Avoid} to \prob{LOP} from \cite{korten2024LOP} proves \cref{thm:main_avoid_ueopl}.

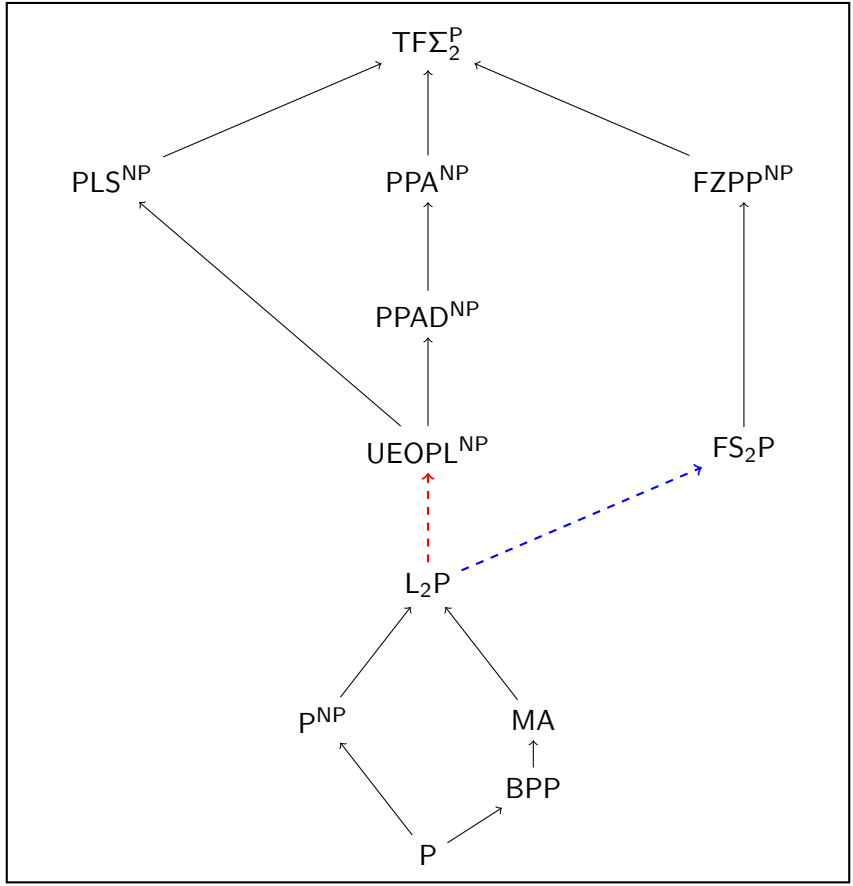
\begin{figure}
    \centering
    \begin{tikzpicture}[x=1.4cm,y=1.8cm]
        nodes\node at (0,-1)     (a)  {$\cc{P}$};
        notes\node at (1, -0.5)  (b)  {$\cc{BPP}$};
        notes\node at (-1, 0)   (c)  {$\cc{P}^{\cc{NP}}$};
        notes\node at (1, 0)    (d)  {$\cc{MA}$};
        notes\node at (0, 1)    (j)  {$\cc{L_2P}$};
        notes\node at (0, 2)    (e)  {$\cc{UEOPL}^{\cc{NP}}$};
        notes\node at (0, 3)    (f)  {$\cc{PPAD}^{\cc{NP}}$};
        notes\node at (3, 2)    (g)  {$\cc{FS_2P}$};
        notes\node at (-3, 4)   (h)  {$\cc{PLS}^{\cc{NP}}$};
        notes\node at (0, 4)   (i)  {$\cc{PPA}^{\cc{NP}}$};
        notes\node at (3, 4)    (k)  {$\cc{FZPP}^{\cc{NP}}$};
        notes\node at (0, 5)    (l)  {$\cc{TF\Sigma_2^P}$};
        
        \draw[->] (a) -- (c);
        \draw[->] (a) -- (b);
        \draw[->] (b) -- (d);
        \draw[dashed, thick, red, ->] (j) -- (e);
        \draw[dashed, thick, blue, ->] (j) -- (g);
        \draw[->] (e) -- (f);
        \draw[->] (f) -- (i);
        \draw[->] (d) -- (j);
        \draw[->] (e) -- (h);
        \draw[->] (g) -- (k);
        \draw[->] (h) -- (l);
        \draw[->] (i) -- (l);
        \draw[->] (c) -- (j);
        \draw[->] (k) -- (l);

        \draw [thick] (-4,-1.2) rectangle +(8,6.5);
    \end{tikzpicture}

    \caption{A Fine-Grained $\cc{TF\Sigma_2^P}$ Version of the Siper-Gacs-Lautermann Theorem. Solid arrows denote inclusions and dotted arrows denote decision to search reductions. The inclusion $\cc{FS_2P} \subseteq \cc{FZPP}^{\NP}$ is due to \cite{cai2007s2p}. $\cc{L_2P}$ which is a class of decision problems was shown to be equivalent to class of problems that are polynomially many-one reducible to \prob{LOP}\cite{korten2024LOP}. The blue dotted arrow indicates the containment of \prob{LOP} $\in$ $\cc{FS_2P}$ as shown in \cite{korten2024LOP}. The red dotted arrow indicates our new containment of \prob{LOP} $\in$ $\cc{UEOPL}^{\NP}$.} 
    \label{fig: inclusions}
\end{figure}


We now compare the new containment in \cref{thm:main_avoid_ueopl} and \cref{thm:main_lop_ueopl_np} to previous results. \cite{chen2024symmetric} and \cite{li2024symmetric} show that \prob{Avoid} is in $\cc{FS_2P}$. In \cite{korten2024LOP} they define $\cc{L_2P}$ as a new syntactic sub-class of $\cc{S_2P}$ and show that $\MA \subseteq \cc{L_2P}$. Moreover, they show that a language $L \in \cc{L_2P}$ is equivalent to $L$ being polynomial time many-one reducible to \prob{LOP}. Thus,  placing \prob{LOP} in a class of search problems also implies that $\cc{L_2P}$ reduces to this class. While \cite{korten2024LOP} places \prob{LOP} in the class $\cc{FS_2P}$, our result places both \prob{LOP} and \prob{Avoid} in the class $\cc{UEOPL}^{\NP}$. In fact, before \cref{thm:main_avoid_ueopl} it was open if \prob{Avoid} or \prob{LOP} was even in $\cc{PLS}^{\NP}$ or $\cc{PPA}^{\NP}$.  A diagram of our new containment is given in \cref{fig: inclusions}.

We now highlight why our new upper bound of \prob{LOP} and \prob{Avoid} is interesting. 
\begin{enumerate}
\item Since \prob{Avoid} is in $\cc{TF\Sigma_2^P}$, it is interesting to see what is the smallest standard TFNP sub-class with access to an $\NP$ oracle that can solve \prob{Avoid}. $\cc{UEOPL}$ is one of the smallest $\cc{TFNP}$ sub-classes, so getting explicit constructions in $\cc{UEOPL}^{\NP}$ would bring us a step closer to the goal of providing explicit constructions in $\cc{FP}^{\NP}$. Moreover, a sub-exponential simulation of $\cc{UEOPL}$ could yield new approaches to getting an $\cc{E}^{\NP}$ construction of rigid matrices and hard truth tables for circuits of size $cn$ for some fixed constant $c > 0$.\footnote{Even constructing circuits of size at least $4n$ would be interesting as none of the current techniques seem to be able to beat the $3.2n$ lower bound.}, both big open problems.

\item The original motivation when defining $\cc{S_2P}$ by \cite{russell1998symmetric} and \cite{canetti1996more} was to find the smallest class within the polynomial hierarchy that captures probabilistic classes like $\BPP$ and $\MA$. Such a class, inherently captures an abundance of solutions (most random strings are good). And while, \prob{Avoid} seems to share this property of abundance of solutions, the results of \cite{korten2024LOP} indicate that the algorithms for \prob{Avoid} \cite{chen2024symmetric, li2024symmetric} isolate a single solution and implicitly solve the more restrictive \prob{LOP} problem, which has a unique solution. In this regard, it seems that the class $\cc{UEOPL}^{\NP}$ which doesn't correspond to any notion of abundance, and has a unique solution captures the behavior of \prob{LOP} better than $\cc{FS_2P}$.

\item We do not understand how the complexity classes $\cc{UEOPL}^{\NP}$ and $\cc{FS_2P}$ relate to each other. Under very strong hardness assumptions for derandomization \cite{klivans1999graph} we have that $\cc{FS_2P} \subseteq \cc{FZPP}^{\NP} \subseteq  \cc{FP}^{\NP} \subseteq \cc{UEOPL}^{\cc{NP}}$. However, unconditionally they seem incomparable.

\end{enumerate}

\subsubsection {On the complexity of \prob{King}}
We now consider the third level of the polynomial hierarchy and look at search problems in $\cc{TF\Sigma_3^P}$. \cite{kleinberg2021total} illustrated two problems: \prob{Shattering} and \prob{King}, that they showed belonged to $\cc{TF\Sigma_3^P}$. In the case of \prob{Shattering}, they gave a non-trivial application of $\cc{\Sigma_2^P}$ oracles to construct a hashing circuit where a collision corresponded to a solution to \prob{Shattering}, thus putting \prob{Shattering} $\in \cc{PPP}^{\cc{\Sigma_2^P}}$. However, in the case of \prob{King}, no such collapse is known. As a result \prob{King} remains the only candidate for a problem in $\cc{TF\Sigma_3^P}$ that doesn't collapse to a smaller $\cc{TFNP}$ sub-class relative to a $\cc{\Sigma_2^P}$ oracle.

We say that a vertex $v$ in a digraph $G$ is a king if every vertex in $G$ can be reached from $v$ by a path of length at most 2. Moreover, we say that $G$ is a tournament if for every pair of distinct vertices $u$, $v$ exactly one of $(u,v)$ or $(v, u)$ is present in $G$. It was shown in \cite{landau1953dominance} that any digraph representing a tournament always has a king. We now define the total problem \prob{King} as introduced by \cite{kleinberg2021total}:

\begin{maindefinition}[\prob{King}]
\label{def: rp-lcp}
In the problem \prob{King}, given as input a circuit $C: [2^n] \times [2^n] \rightarrow \{ 0, 1\}$ representing a digraph, either: 
\begin{enumerate}
    \item 
    find a distinct $x_1, x_2 \in \bit^n$ such that $C(x_1, x_2) = C(x_2, x_1)$,
    \item 
    find an element $k \in [2^n]$ such that for every $a \in [2^n]\setminus \{ k\}$, either $C(k, a) = 1$, or there exists an element $b \in [2^n]\setminus \{ k, a\}$, such that $C(k, b) = 1$ and $C(b, a) = 1$.
\end{enumerate}
\end{maindefinition}

If one tried to make the proof of totality of \prob{King} \cite{landau1953dominance} constructive it seems crucial that one has to count the number of neighbors of a given vertex, thus reducing it to some kind of a generic counting problem. As a result, when discussing the complexity of \prob{King} \cite{kleinberg2021total} make the  following statement:
\begin{center}
    \enquote{Unless $\#\P$ is in the polynomial hierarchy, \prob{King} does not appear to belong in $\cc{PLS}^{\cc{\Sigma_2^P}}$}
\end{center}

Our first result in this section (see \cref{thm: king_pls}) seems to completely counter the expectations of \cite{kleinberg2021total}.

\begin{maintheorem}\label{thm:main_king_pls}
    \prob{King} is in $\cc{PLS}^{\cc{\Sigma_2^P}}$
\end{maintheorem}

To get this result, we are not able to use \cref{thm:main_mu-dsr} directly. However, we note that our proof directly borrows ideas from the $\mu$-downward self-reducible framework. In \cref{lem: king_choices} we identify a very weak downward self-reducible property relating to tournaments that removes one vertex at a time. While we cannot apply \cref{thm:main_mu-dsr} since our $\mu$ here would be exponential, carefully applying ideas from the proof of \cref{thm:main_mu-dsr} allows us to reduce \prob{King} to $\cc{PLS}^{\cc{\Sigma_2^P}}$ directly.


While \cref{thm:main_king_pls} already seems to be counter intuitive, we show perhaps an even more surprising result (see \cref{thm: king_tfzpp}).

\begin{maintheorem}\label{thm:main_zpp_king}
    \prob{King} is in  $\cc{TFZPP}^{\cc{\Sigma_2^P}}$
\end{maintheorem}

To get this result we use the fact that one can sample a $\coNP$ predicate uniformly in $\ZPP^{\cc{\Sigma_2^P}}$\cite{BGP00NPsampler}. We consider the set $W_{u}$ that contains all vertices in the graph that witness that $u$ is not a \prob{King}. We show (\cref{clm: strict_decrease}) that any vertex $v \in W_u$, must have a strictly smaller $W_v$, moreover for a random $v \in W_u$, $W_v$ will be a constant fraction smaller than $W_u$ with high probability (\cref{clm:const_frac}). This leads to a natural algorithm that keeps iteratively sampling through $W_i$ till it becomes empty at which point we find a king.

Thus given an oracle $\cc{\Sigma_2^P}$ we get the following corollary from \cref{thm:main_king_pls} and \cref{thm:main_zpp_king}:

\begin{maincorollary}
    King is in $\cc{PLS}^{\cc{\Sigma_2^P}} \cap \cc{TFZPP}^{\cc{\Sigma_2^P}}$
\end{maincorollary}

We note that this stands in striking contrast to the fact that \prob{King} has no faster randomized algorithm better than $2^n$ in the black-box query model. Therefore, our result shows that a \cc{\Sigma_2^P} oracle likely significantly speeds up algorithms for \prob{King}. In fact, even an \cc{NP} oracle allows us to achieve a significant speedup for \prob{King} (to $\poly(n) 2^n$ time), which requires $\Omega(2^{4n/3})$ in the black-box query model without a \cc{NP} oracle (\cref{thm: king_efficient}).

\subsubsection{Downward Self Reductions in \cc{TFNP}}
We apply \cref{thm:main_mu-dsr} to known \cc{TFNP} problems with the dual purposes of demonstrating how one should use the $\mu$-d.s.r framework and giving alternative proofs of well-known \cc{PLS}/\cc{UEOPL} membership results.

We begin with the P-matrix principal linear complementarity (\prob{P-LCP}) problem. The promise version of \prob{P-LCP} where the input is guaranteed to be a P-matrix\footnote{A P-matrix is a square matrix where every principal minor is positive.} seems to be of the most interest. \cite{fearnley2020unique} showed that this problem is in \cc{UEOPL} in two ways. The first proof reduces \prob{P-LCP} to a different problem \prob{Unique-Sink-Orientation}, which they also show is in $\cc{UEOPL}$. The second proof directly reduces \prob{P-LCP} to \cc{UEOPL} but requires the application of Lemke's algorithm.  However, our proof is significantly shorter and simpler (see \cref{cor:plcp_ueopl}). It only relies on the observation that \prob{P-LCP} is downward self-reducible. Interestingly, and quite unusually for proofs of membership in \cc{TFNP} subclasses, this proof is inherently quite different than the proof that a P-matrix always has a solution to the \prob{Promise-P-LCP} problem.

We then move on to the graph games, which are equivalent to parity games. These types of games were first shown to be in \cc{PLS} in \cite{beckmann2008complexity} via bounded arithmetic. We are able to give an elementary (not using bounded arithmetic) proof of \cc{PLS} membership (see \cref{cor:memdet}). 

We next move to the \prob{Tarksi} problem (see \cref{cor: tarski}). It was shown in \cite{etessami2020tarski} that \prob{Tarski} is in \cc{PLS}. We give an alternative proof through the $\mu$-d.s.r framework that \prob{Tarski} is in \cc{PLS}. Moreover, our proof is a translation of the classic divide and conquer algorithm for finding Tarski fixed-points, and in this way, more intuitive. The proof seems to rely crucially on the notion of $\mu$-d.s.r as opposed to traditional d.s.r. In this instance, the measure $\mu$ captures the size of the lattice that the Tarski instance is defined on. It would be interesting to giving a traditional d.s.r for \prob{Tarksi} or show that one does not exist.

Finally, we cover the \prob{S-Arrival} problem. This problem was first shown to be in \cc{UEOPL} in \cite{fearnley2020unique}. We show that \prob{S-Arrival} is downward self-reducible. This along with the fact that \prob{S-Arrival} is in \cc{TFUP} acts as an alternative proof that \prob{S-Arrival} is in \cc{UEOPL} (see \cref{cor:s_arrival}). While the original proof of \cc{UEOPL} membership is quite simple, our proof highlights the recursive nature of the problem.

\subsubsection{The limits of the d.s.r framework}
Finally, we tackle a problem left open in \cite{harsha2022downward}: is every \cc{PLS}-complete problem traditionally d.s.r? We give a negative answer to this question (unless $\cc{PLS}=\cc{FP}$).
\begin{maintheorem}\label{thm:main_not_all_dsr}
    All \cc{PLS}-complete problems are traditionally d.s.r if and only if $\cc{PLS}=\cc{FP}$.
\end{maintheorem}
\cref{thm:main_not_all_dsr} follows from a simple padding argument. Though not the focus of this work, we note in \cref{obs: np_not_dsr} that a very similar argument shows that all \cc{NP}-complete/\cc{PSPACE}-complete problems are traditionally d.s.r if and only if $\cc{P}=\cc{NP}$/$\cc{P}=\cc{PSPACE}$. Interestingly, we note that the above theorem does not provide a barrier to proving that every \cc{PLS}-complete problem is $\mu$-d.s.r.


Another question asked in \cite{harsha2022downward} is if one can establish a connection between random-self-reducibility and some syntactic \cc{TFNP} subclass (e.g. \cc{PPP}) in the same way that downward self-reducibility implies membership in \cc{PLS}. In \cref{subsec: rsr}, we give some intuition as to why this might be challenging.

\subsection{Related Work}
The notion of downward self-reducibility has a long history and we refer the interested reader to \cite{allender2010new} as a good starting point.

Downward self-reductions serve as one of the main tools in designing search to decision reductions for a plethora of problems. \cite{harsha2022downward} were the first to consider downward self-reducibility in the context of total problems and observe any interesting consequences. They also showed that the \cc{PLS}-complete problem \prob{Iter} is downward self-reducible, thereby showing that a downward self-reducible \cc{TFNP} problem is hard if and only if \cc{PLS} is hard. Along with traditional d.s.r, they also define a more general notion which they refer to as circuit d.s.r. Both these notions are subsumed by $\mu$-d.s.r (\cref{def: mu-dsr}). \cite{harsha2022downward} only used that d.s.r framework to show \cc{PLS} member for \prob{Iter}, \prob{Iter-with-source}, \prob{Sink-of-DAG}, and \prob{Sink-of-DAG-with-source}, all problems closely related to the canonical \cc{PLS}-complete problem \prob{Sink-of-DAG}. We use our framework to show \cc{PLS} and \cc{UEOPL} membership for a wide range of problems. Our framework is more broadly applicable than that of \cite{harsha2022downward}. We are able to use our framework to classify many problems into \cc{PLS^{\Sigma_i^P}}/\cc{UEOPL^{\Sigma_i^P}}, many of which are beyond the scope of \cite{harsha2022downward}.

\cite{bitansky2022ppad} showed another interesting connection between downward self-reducibility and \cc{TFNP}. They showed that the hardness of any \cc{TFNP} problem which is d.s.r and has a randomized batching property is sufficient to show the hardness of \cc{UEOPL}.

\cite{kleinberg2021total} introduced the total function polynomial hierarchy \cc{TF\Sigma_i^P}. They also introduced the idea of raising a \cc{TFNP} subclass (like \cc{PPP}) to a \cc{\Sigma_i^P} oracle to create a \cc{TF\Sigma_{i+1}^P} subclass (like \cc{PPP^{NP}}). They showed that the computational problem associated with the Sauer–Shelah lemma, \prob{Shattering}, is in \cc{PPP^{\Sigma_2^P}}. \cite{kleinberg2021total} also introduced the problem \prob{Avoid}, which has proven fruitful for understanding several areas of complexity theory \cite{korten2022hardest, chen2024symmetric, li2024symmetric, gajulapalli2024oblivious}. Barriers for designing $\cc{FP}$ algorithms \cite{ilango2023indistinguishability}, and $\cc{FNP}$ algorithms \cite{chen2024hardness} for \prob{Avoid} were shown under cryptographic assumptions. In concurrent work, \cite{hirsch2025upper} give an upper bound of $\P^{\cc{prSBP}}$ for the class $\cc{L_2P}$.

Our algorithms for \prob{King} are similar in spirit to the work of \cite{jin2024faster, zhang2024pigeonhole} where they give exponential time but faster than brute-force algorithms for the \cc{TFNP} problem \prob{Pigeonhole Equal Sums}. They do this by exploiting the fact that the \prob{Pigeonhole Equal Sums} is total, unlike the related subset sum problem. Our work therefore contributes to the nascent field of designing algorithms for total search problems.

\subsection{Discussion and Open questions}

We have seen that both \prob{LOP} and \prob{King} are in very small subclasses of \cc{TF\Sigma_i^P}, i.e. $\cc{UEOPL^{\NP}} \cap \cc{TFZPP^{NP}}$\footnote{The containment of \prob{LOP} in \cc{TFZPP^{NP}} follows from approximate counting.}  and $\cc{PLS^{\Sigma_2^P}} \cap \cc{TFZPP^{\Sigma_2^P}}$ respectively. There are two ways to interpret our results: we implicitly capture some kind of combinatorial property of the underlying \cc{TFNP} sub-class via downward self-reductions, or $\Sigma_{i-1}$ oracles are just very powerful.

On one hand, it is possible that the classic \cc{TFNP} subclasses (e.g, \cc{PLS}) seem to capture most of the combinatorial principles used to show totality of \cc{TFNP} problems, moreover the classic \cc{TFNP} subclasses given access to a $\cc{\Sigma_{i-1}^P}$ oracle capture most of the combinatorial principles used to show totality of \cc{TF\Sigma_i^P} problems. In some sense, maybe we do not really encounter new combinatorial principles as we go higher in the total function hierarchy.

On the other hand, do we get a $\cc{UEOPL}^{\NP}$ algorithm for \prob{Avoid} because \prob{Avoid} is somehow intimately related to some kind of end-of-type arguments, or is it because we believe that \prob{Avoid} is in $\cc{FP}^{\NP}$, and we are able to use the $\NP$ oracle in some clever way?

We offer no guess as to what the correct answer is, and even formulating this question formally remains an interesting open question.

Finally, we observe that it is not immediate that \prob{Avoid} even has a downward self-reduction. We are able to get an upper bound on \prob{Avoid} by showing that \prob{Avoid} has a polynomial time many-one reduction to another problem (\prob{LOP}) that does have a downward self-reduction. This raises a new approach of showing \cc{PLS} membership for problems that are not downward self-reducible: try finding an adjacent problem that is downward self-reducible, and to which your original problem reduces to.

Below we list a few open questions that arise from our work.
\begin{enumerate}

        \item Is there a candidate $\cc{TF\Sigma_3^P}$ problem that does not collapse to a sub-class of \cc{TFNP} with access to a $\cc{\Sigma_2^P}$ oracle? In the case of $\cc{TF\Sigma_2^P}$ the only problems not known to collapse further down are \prob{Empty} \cite{kleinberg2021total} and \prob{Short-Choice} \cite{pasarkarshortchoice2023}. Do these problems also admit further collapses?
    
    \item Can the d.s.r framework be used to show that the \cc{TF\Sigma_3^P} problem \prob{Shattering} \cite{kleinberg2021total} is in \cc{PLS^{\Sigma_2^P}}? In particular, the proof of the Sauer-Shelah lemma, the combinatorial principle behind \prob{Shattering}, employs a divide-and-conquer argument which seems closely related to a downward self-reduction.
    
    \item Does a non-adaptive downward self-reduction for a \cc{TF\Sigma_i^P} problem imply membership in a smaller class than \cc{PLS^{\Sigma_{i-1}^P}}? This is a particularly important question since our downward self-reductions for graph games (\cref{subsubsec: graph-games}), \prob{P-LCP} (\cref{subsubsec: p-lcp}), and \prob{LOP} (\cref{sec:lop}) are all non-adaptive.

    \item Does $\mu$-d.s.r imply traditional d.s.r?

    \item Aldous' algorithm \cite{aldous1983minimization} gives us a randomized procedure for $\prob{Sink-of-DAG}$ for an input $S: \{ 0, 1\}^n \rightarrow \{ 0, 1\}^n, V: \{ 0, 1\}^n \rightarrow \{ 0, 1\}^n$ that runs in expected time $\poly(n) 2^{n/2}$. \cite{fearnley2020unique} and \cite{gartner2018arrival} gave $\poly(n) 2^{n/2}$ time randomized algorithms for \prob{P-LCP} and \prob{S-Arrival} by showing there exist fine grained reductions from those problems to \prob{Sink-of-DAG}, and then applying Aldous' algorithm. Our proofs using $\mu$-d.s.r on the other hand incur a large blowup in instance size (except that of \prob{King}). Is there any way to apply Aldous' algorithm to achieve a speedup for problems shown to be in $\cc{PLS^{\Sigma_i^P}}$ via $\mu$-d.s.r?


    \item What conditions can we put on a $\mu$-d.s.r to show membership in \cc{TFNP} classes between \cc{PLS} and \cc{UEOPL}? For example, how should we extend the $\mu$-d.s.r framework in a way that enables us to show that \prob{Tarski} is in \cc{CLS}?

\end{enumerate}

\section{Preliminaries}

For a set $S \subseteq S_1 \times \dots \times S_n$, we let $\pi_i(S)$ denote the projection of $S$ onto its $i^{\text{th}}$ coordinate. In particular, $\pi_i(S) = \{ e_i : (e_1, \dots, e_i, \dots, e_n) \in S \}$.

\subsection{The total function hierarchy}
We now formally define a search problem and problems in the total function polynomial hierarchy, i.e. $\cc{TF\Sigma_i^P}$.
\begin{definition}
    A search problem is a binary relation $\mathcal{R} \subseteq \{ 0, 1\}^* \times \{ 0, 1\}^*$ where we say that a $y$ is a solution to $x$ if and only if $(x, y) \in \mathcal{R}$.
\end{definition}

\begin{definition}[\cite{kleinberg2021total}]
    A relation $\mathcal{R}$ is in \cc{TF\Sigma_i^P} if the following conditions hold:
    \begin{enumerate}
        \item $\mathcal{R}$ is total: for all $x \in \{ 0, 1\}^*$, there exists $y \in \{ 0, 1\}^*$ such that $(x, y) \in \mathcal{R}$.
        \item $\mathcal{R}$ is polynomial: For all $(x, y) \in \mathcal{R}$, $|y| \leq \poly(|x|)$.
        \item There exists a polynomial time Turing machine $M$ and polynomials $p_1(n), \dots, p_{i-1}(n)$ such that 
        \[(x, y) \in R \iff \forall z_1 \in \{ 0, 1\}^{p_1(|x|)} \exists z_2 \in \{ 0, 1\}^{p_2(|x|)} \forall z_3 \in \{ 0, 1\}^{p_3(|x|)} \dots V(x, y, z_1, \dots, z_{i-1})\text{ accepts.}\]
    \end{enumerate}
\end{definition}

\begin{definition}
    A relation $\mathcal{R}$ is in \cc{TFU\Sigma_i^P} if it is in \cc{TF\Sigma_i^P} and for all $x \in \{ 0, 1\}^*$, there exists exactly one $y$ such that $(x, y) \in \mathcal{R}$.
\end{definition}
\cc{TF\Sigma_1^P} and \cc{TFU\Sigma_1^P} are also referred to as \cc{TFNP} and \cc{TFUP} respectively. Informally, \cc{TFNP} is the set of search problems where a solution always exists and is efficiently verifiable (is in \cc{FNP}, the search analogue of \cc{NP}). \cc{TF\Sigma_i^P} is a generalization of this notion higher in the polynomial hierarchy, where the verifier need not be polynomial time.

With the notion of search problem pinned down, one can begin to talk about reductions between search problems. The norm when reducing between search problems is to restrict oneself to many-to-one reductions. The following definition formalizes the intuitive idea that a reduction from $\mathcal{R}$ to $\mathcal{Q}$ should take as input an instance $x$ of $\mathcal{R}$, transform it to an instance $f(x)$ of $\mathcal{Q}$, get back an answer $y$ to that instance, and transform that to $g(y)$, an answer for instance $x$ of $\mathcal{R}$.
\begin{definition}
    \label{def: reduction}
    Let $\mathcal{R}, \mathcal{Q}$ be two search problems. A many-to-one reduction from $\mathcal{R}$ to $\mathcal{Q}$ is defined as two polynomial time computable functions $f, g$ such that for all $x \in \pi_1(\mathcal{R}), y \in \{ 0, 1\}^*$, the following holds.
    \[ (x, g(y)) \in \mathcal{R} \impliedby (f(x), y) \in \mathcal{Q} \]
\end{definition}

We note that \cref{def: reduction} is slightly different from the notion of a reduction between \cc{TFNP} problems since we only quantify over all $x \in \pi_1(\mathcal{R})$ rather than all $x \in \{ 0, 1\}^*$. We do this since we will often work with problems where the inputs to the search problem are restricted and there exists some $x$ such that $x \notin \pi_1(\mathcal{R})$. One should think of these relations as promise problems. Occasionally, it will be helpful to make these problems total in a trivial way, so we define the following.
\begin{definition}\label{def:relation_total}
    Let $\mathcal{R}$ be a relation and let $A = \{ (x, \bot) : x \notin \pi_1(\mathcal{R})\}$. We define the completion of $\mathcal{R}$ as $\overline{\mathcal{R}} := \mathcal{R} \cup A$.
\end{definition}

We now define what it means for an (oracle) algorithm to solve a search problem $\mathcal{R}$ in polynomial time.
\begin{definition}
    We say $\mathcal{R}$ has a polynomial time algorithm $\mathcal{A}$ if $\mathcal{A}$ runs in polynomial time for all inputs in $\{ 0, 1\}^*$ and for all $x \in \pi_1(\mathcal{R})$, $(x, \mathcal{A}(x)) \in \mathcal{R}$. 
\end{definition}

\begin{definition}
    Let $\mathcal{Q}$ be a total search problem. We say $\mathcal{R}$ has a polynomial time algorithm $\mathcal{A}^{\mathcal{Q}}$ if $\mathcal{A}^{\mathcal{Q}}$ runs in polynomial time and for all $x \in \pi_1(\mathcal{R})$ and all instantiations of the oracle $\mathcal{Q}$, $(x, \mathcal{A}^{\mathcal{Q}}(x)) \in \mathcal{R}$.
\end{definition}




\subsection{TFNP Sub-problems}
Finally, we will make extensive use of the following \cc{TFNP} subclasses. For any \cc{TFNP} subclass \cc{C} has a circuit as part of its input, we define $\cc{C^{\Sigma_i^P}}$ as the same problem as \cc{C} except that the input circuit to \cc{C} may include $\cc{\Sigma_i^P}$ oracle gates. It is not hard to confirm that $\cc{C^{\Sigma_i^P}} \subseteq \cc{TF \Sigma_{i+1}^P}$.

\begin{definition}[\cc{PLS} and \prob{Sink-of-DAG}]
    \cc{PLS} is defined as all problems which are many-to-one reducible to \prob{Sink-of-DAG}, which is defined as follows. Given $\poly(n)$ size circuits $S: [2^n] \rightarrow [2^n]$ (successor circuit) and $V: [2^n] \rightarrow [2^n]$ (value circuit), find $v$ such that $S(v) \neq v$ and either $S(S(v)) = S(v)$ or $V(S(v)) \leq V(v)$.
\end{definition}
One should think of \prob{Sink-of-DAG} as a solution to a gradient ascent problem. There are two circuits, a successor circuit and a value circuit. At every point $v$ in the space, we hope that the successor function $S$ leads to a point which has a higher value $(V(S(v)) > V(v))$. A solution to \prob{Sink-of-DAG} is a point such that this condition is violated $V(S(v)) \leq V(v)$, or one which acts as (a predecessor of) a sink in the gradient ascent process, where $S(v) \neq v$ but $S(S(v)) = S(v)$.


We will also make use of the following promise problem. \prob{Sink-of-Verifiable-Line} is very similar to \prob{Sink-of-DAG} except that we also require a verifier which can tell us if a vertex $x$ is the $i^{\text{th}}$ vertex visited by repeatedly applying the successor function. Note that \prob{Sink-of-Verifiable-Line} is a promise problem since we have no way to confirm if the verifier meets this condition.
\begin{definition}[\prob{Sink-of-Verifiable-Line}] Given a successor circuit $S : \{0, 1\}^n \to \{0, 1\}^n$,  a source $s \in \{0, 1\}^n$, a target index $T \in [2^n]$, and a verifier circuit 
$W : \{0, 1\}^n \times [T] \to \{0, 1\}$ with the guarantee that for $(x, i) \in \{0, 1\}^n \times [T]$, $W(x, i) = 1$ if and only if $x = S^{i-1}(s)$, find the string $v \in \{0, 1\}^n$ such that $W(v, T) = 1$.
\end{definition}




Finally, although we will not make use of the complete problem for the \cc{UEOPL}, we present it here for the sake of completeness. See \cite{fearnley2020unique} for an explanation of why the following is a reasonable definition.

\begin{definition}[\prob{Unique-End-of-Potential-Line}]
\prob{Unique-End-of-Potential-Line} is defined as follows. Given two $\poly(n)$ size circuits $S, P : \bit^n \rightarrow \bit^n$, such that $P(0^n) = 0^n \neq S(0^n)$, and a value circuit $V:\bit^n \rightarrow \{0, 1, \hdots, 2^{m-1}\}$ where $V(0^n) = 0$. Find one of the following:
\begin{itemize}
    \item A point $x \in \bit^n$ such that $P(S(x)) \neq x$.

    \item A point $x \in \bit^n$ such that $x \neq S(x), P(S(X)) = x$ and $V(S(x)) - V(x) \leq 0$.

    \item A point $x \in \bit^n$ such that $S(P(x)) \neq x \neq 0^n$. 

    \item Two points $x, y \in \bit^n$, such that $x \neq y$, $x \neq S(x)$, $y \neq S(y)$, and either $V(x) = V(y)$ or $V(x) < V(y) < S(x)$.
\end{itemize}
\cc{UEOPL} is defined as all problems which are many-to-one reducible to \prob{Unique-End-of-Potential-Line}.
\end{definition}

\begin{lemma}[\cite{fearnley2020unique}]
    \label{lem: sovl_to_ueopl}
    \prob{Sink-of-Verifiable-Line} reduces to \prob{Unique-End-of-Potential-Line} under relativizing reductions.
\end{lemma}





\section{Extending Downward self-reducibility}\label{sec:mu_dsr}

\subsection{\texorpdfstring{$\boldsymbol{\mu}$}{mu}-Downward self-reducibility}
We begin by defining the types of problems for which we hope to show downward self-reducibility: Efficiently verifiable promise problems.

\begin{definition}
    \label{def: promiseTFSigma_i}
    We say that a search problem $\mathcal{R} \subseteq \{ 0, 1\}^* \times \{ 0, 1\}^*$ is in the class \cc{PromiseF\Sigma_i^P} if it satisfies the following conditions.
    \begin{enumerate}
        \item
        For all $x \in \pi_1(\mathcal{R})$, and for all $y$ s.t. $(x, y) \in \mathcal{R}$, $|y| \leq \poly(|x|)$.
        \item 
        There exists a polynomial time verifier $V$ and polynomials $p_1(n), \dots, p_{i-1}(n)$ such that the following holds: for all $x \in \pi_1(R)$, 
        \[(x, y) \in R \iff \forall z_1 \in \{ 0, 1\}^{p_1(|x|)} \exists z_2 \in \{ 0, 1\}^{p_2(|x|)} \forall z_3 \in \{ 0, 1\}^{p_3(|x|)} \dots V(x, y, z_1, \dots, z_{i-1})\text{ accepts.}\]
    \end{enumerate}
\end{definition}

 \begin{definition}
     We say a search problem $\mathcal{R} \subseteq \{ 0, 1\}^* \times \{ 0, 1\}^*$ has unique solutions if for all $x \in \{ 0, 1\}^{*}$, there exists at most one $y$ such that $(x, y) \in \mathcal{R}$.
 \end{definition}

Perhaps the most helpful way to think of \cc{PromiseF\Sigma_i^P} is as a version of \cc{TF\Sigma_i^P} where some input instances are disallowed. In particular, we are only concerned with the instances $\pi_1(\mathcal{R})$\footnote{$\pi_1(\mathcal{R})$ is the set of instances that have a solution, for example, in the case of \prob{FSAT}: $\pi_1(\mathcal{R})$ is the set of satisfiable formulas.}. For a problem $\mathcal{R}$, we will often say that $x$ meets the promise of \cc{TF\Sigma_i^P} if $x \in \pi_1(\mathcal{R})$. Some examples of \cc{PromiseF\Sigma_i^P} problems to have in mind are \prob{Promise-FSAT} and \prob{Promise-Unique-FSAT}.


\begin{definition}
    \prob{Promise-FSAT}: $(\varphi, y) \in$ \prob{Promise-FSAT} if $\varphi$ is a satisfiable boolean formula, and $y$ is a satisfying assignment for $\varphi$.
\end{definition}

\begin{definition}
    \prob{Promise-Unique-FSAT}: $(\varphi, y) \in$ \prob{Promise-Unique-FSAT} if $\varphi$ is a satisfiable boolean formula with a unique satisfying assignment, and $y$ corresponds to the unique satisfying assignment for $\varphi$.
\end{definition}


We now define $\mu$-downward self-reducibility. Intuitively, a problem $\mathcal{R}$ is $\mu$-d.s.r if there exists some function $\mu$ over instances of $\mathcal{R}$ such that there exists a self-reduction for $\mathcal{R}$ which calls its $\mathcal{R}$ oracle on instances which have strictly smaller values for $\mu$. A traditionally d.s.r problem is $\mu$-d.s.r for $\mu(x) = |x|$. In some sense, $\mathcal{R}$ being $\mu$-d.s.r means that there is a recursive algorithm where $\mu$ decreases at each recursive call. In this way, we believe $\mu$-d.s.r more fully captures recursive algorithms over traditional d.s.r.
\begin{definition}
    \label{def: mu-dsr}
    Let $\mu:\bit^* \rightarrow \N$ be a function. A \cc{PromiseF\Sigma_i^P} problem $\mathcal{R}$ is (randomized) $\cc{\Sigma_{i-1}^P}$-$\mu$-downward-self-reducible ($\mu$-d.s.r.) if there is a (randomized) polynomial-time oracle algorithm $\mathcal{A}^{\cc{\Sigma_{i-1}^P}, \overline{\mathcal{R}}}$ for $\mathcal{R}$ that on input $x \in \pi_1(\mathcal{R})$, makes queries to a $\mathcal{R}$-oracle on instances $y$ such that $y$ satisfies the promise of $\mathcal{R}$, $|y| \leq |x| + \poly(\mu(x))$, and $\mu(y) < \mu(x)$. We simply say that such a problem is $\mu$-downward-self-reducible when $\cc{\Sigma_{i-1}^P}$ is clear from context.
\end{definition}

\gnote{This is new. Please check it.}
\znote{made it into bullet points.}

\begin{remark}
    We make a few comments about the definition above.
    \begin{itemize}
        \item It is implicit in the \cref{def: mu-dsr} that for any instance $x$ where $\mu(x) = 0$, the reduction algorithm $\mathcal{A}^{\cc{\Sigma_{i-1}^P}, \overline{\mathcal{R}}}$ straightaway solves $x$ since it cannot make any query of smaller $\mu$.

        \item To be fully formal, \cref{def: mu-dsr} specifies that $\mathcal{R}$ is $\cc{\Sigma_{i-1}^P}$-$\mu$-downward-self-reducible rather than just saying it is $\mu$-downward-self-reducible to eliminate trivial cases like the following. $\prob{Promise-FSAT}$ is a $\cc{TF\Sigma_2^P}$ problem since it is a $\cc{FNP}$ problem. Furthermore, $\prob{Promise-FSAT}$ is $\cc{\Sigma_1^P}$-$\mu$-d.s.r since one can use the $\cc{\Sigma_1^P}$ oracle to find a solution. So in some sense, $\prob{Promise-FSAT}$ is $\mu$-d.s.r. However, this clearly misses the point that we wish for the d.s.r to have access to less computational power than an oracle which simply solves the problem in question. We therefore parameterize the d.s.r by the oracle it is allowed access to ($\cc{\Sigma_{i-1}^P}$-$\mu$-downward-self-reducible). However, for the rest of this work, it will be clear what $\cc{\Sigma_{i-1}^P}$ should be and we therefore omit it and simply write $\mu$-d.s.r to avoid unnecessary notation.

        \item In the definition of $\mu$-d.s.r, if $x$ meets the promise of $\mathcal{R}$, then all the queries which $\mathcal{A}^{\cc{\Sigma_{i-1}^P}, \overline{\mathcal{R}}}$\footnote{$\overline{\mathcal{R}}$ is the extension of $\mathcal{R}$ that makes it total (see \cref{def:relation_total}).} makes to its $\mathcal{R}$ oracle, call them $y_i$, must satisfy $y_i \in \pi_1(R)$. In other words, if $x$ satisfies the promise of $\mathcal{R}$, the $\mu$-d.s.r for $\mathcal{R}$ can only make queries to its oracle which satisfy the promise of $\mathcal{R}$. Notice also that by our definition of $\mu$-d.s.r, we do not have to worry about the behavior of the oracle for $\mathcal{A}$ on inputs which do not satisfy the promise of $\mathcal{R}$. One should imagine these inputs as causing the oracle to output $\bot$. However we wish to emphasize that this is just for definitional convenience and that $\mathcal{A}$ can never use such information, since by definition, it never queries its oracle on such an input.

        For intuition, let us consider what a $\mu$-d.s.r for \prob{Promise-FSAT} would look like. This would be an algorithm which on a satisfiable SAT instance $x$ would query \emph{only} satisfiable SAT instances $y_1, \dots, y_{\poly(n)}$ such that each $y_i$ is smaller than $x$ (under $\mu$), receives back satisfying assignments $a_1, \dots, a_{\poly(n)}$, and then uses those to construct a satisfying assignment for $x$.
    \end{itemize}
\end{remark}


We now prove one of our main theorems relating $\mu$-d.s.r to \cc{PLS}/\cc{UEOPL}. The proof uses essentially the same strategy as \cite{harsha2022downward} of creating a \prob{Sink-of-DAG} instance where nodes represent possible stack traces of the natural recursive algorithm for solving our $\mu$-d.s.r problem.

\begin{theorem}
\label{thm:mu-dsr}
    Let $\mathcal{R}$ be a search problem in \cc{PromiseF\Sigma_i^P} which has a (randomized) $\mu$-d.s.r. If $\mu(x) \leq p(|x|)$ for a polynomial $p$, then there is a (randomized) reduction from $\mathcal{R}$ to $\PLS^{\cc{\Sigma_{i-1}^P}}$. Furthermore, if $\mathcal{R}$ has unique solutions, then there is a (randomized) reduction from $\mathcal{R}$ to $\cc{UEOPL}^{\cc{\Sigma_{i-1}^P}}$.
\end{theorem}

\znote{Please double check if the updated proof makes sense. Especially second paragraph}

\begin{proof}
    We start by proving the simple case where the $\mu$-d.s.r. algorithm is deterministic and $i = 1$, i.e. $\mathcal{R} \in \cc{\cc{PromiseF\Sigma_1^P}}$. 
    If the following reduction breaks at any point or runs for too long, the reduction simply outputs $\bot$. This will be fine since we will show that the reduction works whenever $x$ satisfies the promise of $\mathcal{R}$.
    
    Let $A'$ be the polynomial-time oracle algorithm for $\mathcal{R}$. Without loss of generality, let us assume that on input $x\in \bit^n$, $A'$ makes at most $q(n)$ queries, each of size bounded by $n+\ell(\mu(x))$ for polynomials $q$ and $\ell$. Let $A$ be an (exponential-time) depth-first recursive algorithm that simulates $A'$ and recursively calls $A'$ whenever $A'$ makes an oracle query. Due to the strictly decreasing property of $\mu(x)$ in the $\mu$-d.s.r. definition, the maximum depth of recursion is $\mu(x) \leq p(n)$ and all instances have size bounded by $s:= n + \mu(x)\ell(\mu(x))$. $A'$ on input $x$ computes depth $d := \mu(x)$ and width $w := q(s)$. Without loss of generality by introducing dummy oracle queries, we enforce $A'$ to make \emph{exactly} $w$ oracle queries at each simulation of $A'$ and have recursive depth \emph{exactly} $d$. In other words, the computational graph of $A$ would be a perfect $w$-ary tree of height $d$. For simplicity, let us further assume that the solutions have the same size as the instances, which can be done by padding $\mathcal{R}$ in the first place.

    On input $x \in \bit^n$, we will construct a $\prob{Sink-of-DAG}$ instance as follows. Each vertex is represented by a table $t[\cdot, \cdot]$ of $d + 1$ rows and $w$ columns. Each entry $t[i,j] \in \bit^{s}\times \bit^{s}$ consists of an instance-solution pair and may take one of the following forms:
    \begin{enumerate}
        \item $(\xi, \bot)$ of an $\mathcal{R}$-instance $\xi$ and $\bot$ indicating that its solution is yet to be found.
        \item $(\xi, \gamma)$ of an $\mathcal{R}$-instance $\xi$ and its purported solution $\gamma$.
        \item $(\bot, \bot)$ indicating that there is no instance.
    \end{enumerate}

    Intuitively, $t[\cdot, \cdot]$ is the stack content of the recursive algorithm $A$. On input $x$, the source vertex $t_0$ is defined by setting $t_0[0,1] = (x,\bot)$ and $t_0[i,j] = (\bot,\bot)$ for all other $i,j$. 

    \paragraph{Validity} A vertex $t$ is valid if and only if it passes the following validity test: Let $i^*$ be the smallest index such that the row $t[i^*,\cdot]$ consists of only $(\bot, \bot)$. For $i \in [p(n)]$, let $j_i^*$ be the smallest index such that $t[i,j_i^*]$ takes the form of $(\xi, \bot)$.
    \begin{itemize}
        \item All rows $t[i,\cdot]$ for $i \geq i^*$ consist of only $(\bot, \bot)$. For any $i \in [p(n)]$ and $j \geq j_i^*$, $t[i,j] = (\bot, \bot)$.
        \item For any $t[i,j]$ taking the form $(\xi, \gamma)$, $\gamma$ is a valid solution for $\xi$. We check this using the verifier for $\mathcal{R}$.
        \item For any $t[i,j] \neq (\bot,\bot)$, consider simulating $A'$ on the input instance stored in $t[i-1, j^*_{i-1}]$, and verify that the first $j-1$ query-solution pairs are in $t[i,1], \ldots, t[i,j-1]$, and $t[i,j]$ is the $j$-th query made by $A'$.
    \end{itemize}

    \paragraph{Successor} On vertex $t$, the successor circuit $S$ behaves as follows.
    \begin{itemize}
        \item If $t$ is invalid, then $S(t) = t$. These are isolated vertices.
        \item If $t[0,1] = (x,y)$ and $y\neq \bot$ is a valid solution to $x$, then $S(t) = t$. 
        \item Otherwise, we construct successor vertex $t'$ as follows. Let $i$ be the largest index such that the $i$-th row $t[i,\cdot]$ consists of an entry of the form $(\xi,\bot)$. Suppose that the $(i+1)$-th row $t[i+1, \cdot]$ has $j$ query-solution pairs. If $A'$ on input $\xi$ and these $j$ query-solution pairs makes a $(j+1)$-th query $\xi'$, then we set $t[i+1, j+1]:=(\xi', \bot)$. Otherwise $A'$ would have found a solution $\gamma$ for $\xi$. In this case, we replace $(\xi, \bot)$ by $(\xi, \gamma)$ and set $t[i+1,\cdot]$ to all $(\bot, \bot)$.
    \end{itemize}

    \paragraph{Potential} For invalid $t$, $V(t) = 0$. For valid vertex $t$, the potential circuit $V$ returns the \emph{exact} number of steps taken for the depth-first recursive algorithm $A$ to reach the state $t$. This can be easily computed as follows: Go through the table $t$, for each $t[i,j] = (\xi, \gamma)$, add $\sum_{k = i}^{d} 2w^{d-k}$; for each $t[i,j] = (\xi, \bot)$, add $1$. 
    We provide a quick verification that for a non-sink vertex $t$, $V(S(t)) = V(t) + 1$:
    \begin{itemize}
        \item If $S(t)$ writes an entry from $(\bot, \bot)$ to $(\xi, \bot)$, the potential increases by exactly 1.
        \item If $S(t)$ writes an entry $t[i,j]$ from $(\xi, \bot)$ to $(\xi, \gamma)$ for $i = d$ (i.e. at the leaf level), the potential changes by 2 - 1 = 1.
        \item If $S(t)$ writes an entry $t[i,j]$ from $(\xi, \bot)$ to $(\xi, \gamma)$ for $i < d$ and erase the row $t[i+1, \cdot]$ to $(\bot, \bot)$, by assumption all $w$ entries in $t[t+1, \cdot]$ were of the form $(\xi', \gamma')$. Hence the potential changes by $\sum_{k = i}^{d} 2w^{d-k} - w \cdot \sum_{k = {i+1}}^{d} 2w^{d-k} - 1 = 1$.
    \end{itemize}

    \paragraph{} Correctness follows from that the only non-isolated vertices without a proper successor correspond to the final configurations of $A$ with the solution to $x$ written in $t[0,1]$.

    \paragraph{Randomized Reduction} If the $\mu$-d.s.r. algorithm $A'$ is randomized, without loss of generality assume that $A'$ uses $r(n)$ bits of randomness and succeeds with probability at least $1 - 1/\nu(n)$ for polynomials $r(\cdot)$ and $\nu(\cdot)$. 

    Since $\mathcal{R} \in \cc{PromiseF\Sigma_1^P}$, we could verify the validity of the solution. Hence, we start by amplifying the success probability to $1 - \frac{1}{\nu(n)^{n+s}}$ by repeating $A'$ $(n+s)$ times in parallel with $(n+s)r(n)$ bits of randomness. By union bound over at most $2^{s}$ instances of size $s$, the success probability remains exponentially high.

    The randomized reduction to $\PLS$ is as follows. Sample $(n+s)r(n)$ bits of randomness and hardcode them into the $\PLS$ instance. With the randomness fixed, we could apply the reduction from above to obtain the $\PLS$ instance. 

    \paragraph{Relativization} The proof above fully relativizes in the following sense: the \emph{only} property of $\cc{PromiseF\Sigma_1^P}$ used is that given an instance-solution pair $(\xi, \gamma)$, one could verify if $\gamma$ is a solution for $\xi$ in polynomial time. Now if $\mathcal{R}$ is a $\cc{PromiseF\Sigma_i^P}$ problem, one could use a $\cc{\Sigma_{i-1}^P}$ oracle to verify the validity of a purported solution. Moreover, the d.s.r. algorithm $A'$ (with $\cc{\Sigma_{i-1}^P}$ oracle) can be simulated by a circuit with $\cc{\Sigma_{i-1}^P}$ oracle gates.
    
    \paragraph{Unique Solutions} We reduce to \prob{Sink-Of-Verifiable-Line} and the appeal to the fact that \prob{Sink-Of-Verifiable-Line} reduces to \cc{UEOPL} (\cref{lem: sovl_to_ueopl}). The reduction to \prob{Sink-Of-Verifiable-Line} uses the exact same $S$ as previously but creates $W$ as $W(t, i) = 1$ if and only if $t$ is valid and $V(t) = i$, sets $T = \sum_{k = 0}^{d} 2w^{d-k}$ and then calls \prob{Sink-Of-Verifiable-Line} to get an answer $t'$. The reduction then outputs the solution stored in $t'[0, 1]$. We see that $W$ is a valid verifier for \prob{Sink-Of-Verifiable-Line}. If $S^{i-1}(s) = t$, then $W(t, i) = 1$ because $t$ must be valid and $V(t) = i$ since $V$ increased by exactly one at each application of $S$. If $W(t, i) = 1$, then $S^{i-1}(s) = t$. This follows from the fact that since $\mathcal{R}$ has unique solutions, for every $i$, there exists exactly one valid table $t$ such that $V(t) = i$. Therefore, $W$ satisfies the promise required by \prob{Sink-Of-Verifiable-Line}. $S^{T-1}(s) = t'$ is clearly the $(T-1)^{\text{th}}$ state of the recursive algorithm $A$ and therefore clearly contains the solution for $x$, which the reduction correctly outputs.
\end{proof}

\begin{corollary}
    \label{cor: promise_fsat}
    If \prob{PromiseFSAT} is $\mu$-d.s.r for polynomially bounded $\mu$, then \prob{SAT} reduces to \cc{PLS}. If \prob{Promise-Unique-FSAT} is $\mu$-d.s.r for polynomially bounded $\mu$, then \prob{SAT} reduces to \cc{UEOPL}.
\end{corollary}
\begin{proof}
    Say \prob{PromiseFSAT} were $\mu$-d.s.r for polynomially bounded $\mu$. That would imply a reduction $(f, g)$ from \prob{PromiseFSAT} to \cc{PLS} by \cref{thm:mu-dsr}. We can turn this into a reduction from \prob{SAT} to \cc{PLS} as follows. For an input $x$ to \prob{SAT}, the reduction computes $f(x)$ (notice that by definition, $f$ must terminate in polynomial time and output something even if $x$ does not meet the promise of \prob{PromiseFSAT}) and feeds $f(x)$ to the \cc{PLS} oracle to get back $y$. The reduction then computes $g(y)$ (notice that by definition, $g$ must terminate in polynomial time and output something). If $g(y)$ is a satisfying assignment for $x$, the reduction outputs SAT and otherwise outputs UNSAT. Correctness of the reduction follows from the correctness of the reduction $f, g$ on satisfiable instances. Essentially the same proof holds for \prob{Promise-Unique-FSAT}, but \cc{PLS} is replaced by \cc{UEOPL}.
\end{proof}

\begin{corollary}
    \label{cor: mu-d.s.r}
    Let $\mu:\bit^* \rightarrow \N$ be a function bounded by some polynomial. Any \cc{TF\Sigma_i^P} problem which is $\mu$-d.s.r with a \cc{\Sigma_{i-1}^P} oracle is in \cc{PLS^{\Sigma_{i-1}^P}}. Any \cc{TFU\Sigma_i^P} problem which is $\mu$-d.s.r with a \cc{\Sigma_{i-1}^P} oracle is in \cc{UEOPL^{\Sigma_{i-1}^P}}. 
\end{corollary}
\begin{proof}
    Notice that any $\mu$-d.s.r for a $\prob{A} \in \cc{TF\Sigma_i^P}$ problem is promise-preserving since \prob{A} is total. This lets us apply \cref{thm:mu-dsr} to achieve the desired result.
\end{proof}

\subsection{\texorpdfstring{$\boldsymbol{\mu}$}{mu}-Downward self-reducibility with essentially unique solutions}
\cref{cor: mu-d.s.r} tells us that if a $\cc{TF\Sigma_i^P}$ problem has unique solutions and is $\mu$-d.s.r for polynomially bounded $\mu$, then it is in \cc{UEOPL^{\Sigma_{i-1}^P}}. However, \cite{korten2024LOP} observed that for $\cc{TF\Sigma_2^P}$, unique solutions are sometimes not the correct notion to work with. A good example to keep in mind when thinking about essentially unique solutions is the problem \prob{Empty}.
\begin{definition}
    The problem \prob{Empty} is defined as follows. Given a $\poly(n)$ size circuit $C: [2^n - 1] \rightarrow [2^n]$, output $y$ such that $y \notin \text{range$(C)$}$.
\end{definition}

Although \prob{Empty} technically does not have unique solutions (consider $C(x) = 0$), it seems rather close. In particular, it is easy to verify with a \cc{NP} oracle if we are in the case when \prob{Empty} has a unique solution ($C$ is injective). \cite{korten2024LOP} define the notion of essentially unique solutions to capture these types of problems. The following is a generalization of the notion of essentially unique solutions from \cc{TF\Sigma_2^P} to \cc{TF\Sigma_i^P} for any $i \geq 2$.

\begin{definition}[Essentially Unique Solutions \cite{korten2024LOP}]
    \label{def: essentially_unique}
    We say that a total search problem $\mathcal{R} \in \cc{TF\Sigma_i^P}$ ($i \geq 2$) has essentially unique solutions if there exist two verifiers $V_1, V_2$ such that the following hold.
    \begin{enumerate}
        \item $V_1$ is testable in polynomial time with an \cc{\Sigma_{i-2}^P} oracle and $V_2$ is testable in polynomial time with an \cc{\Sigma_{i-1}^P} oracle.
        \item For all $x$, either there exists a $y$ such that $V_1(x, y) = 1$ and $(x, y) \in \mathcal{R}$, or else there exists a unique $y$ such that $V_2(x, y) = 1$ and $(x, y) \in \mathcal{R}$.
    \end{enumerate}
\end{definition}

We will show that if a $\cc{TF\Sigma_i^P}$ problem has essentially unique solutions and is $\mu$-d.s.r for polynomially bounded $\mu$, then it has a (randomized) reduction to $\cc{UEOPL}^{\cc{\Sigma_{i-1}^P}}$. To do so, for any search problem $\mathcal{R}$ that has essentially unique solutions, we define $\mathcal{R}_u$, a version of $\mathcal{R}$ that has unique solutions.

\begin{definition}
    For any relation $\mathcal{R}$ with essentially unique solutions and verifiers $V_1, V_2$, we define $\mathcal{R}_u$ as follows. 
    \begin{enumerate}
        \item For any $x \in \pi_1(\mathcal{R})$ such that there exists $y$ where $V_1(x, y) = 1$, $(x, y) \in \mathcal{R}_u$ if and only if $y$ is the lexicographically smallest $y$ where $V_1(x,y) = 1$. 
        \item For any $x \in \pi_1(\mathcal{R})$ such that $V_1(x, y) = 0$ for all $y$, $(x, y) \in \mathcal{R}_u$ if and only if $(x, y) \in \mathcal{R}$.
    \end{enumerate}
\end{definition}

Note that if $\mathcal{R} \in \cc{TF\Sigma_i^P}$, then $\mathcal{R}_u \in \cc{TF\Sigma_i^P}$. We now show that $\mathcal{R}$ reduces to $\mathcal{R}_u$, and that if $\mathcal{R}$ is $\mu$-d.s.r, then $\mathcal{R}_u$ is also $\mu$-d.s.r. The combination of these two facts will be sufficient to show $\cc{UEOPL^{\Sigma_{i-1}^P}}$ membership of $\mu$-d.s.r \cc{TF\Sigma_i^P} problems with essentially unique solutions (assuming $\mu$ is polynomially bounded of course).

\begin{lemma} \label{lem:R-to-Ru}
    Let $\mathcal{R} \in \cc{TF\Sigma_i^P}$ be any relation with essentially unique solutions. Then, there is a reduction from $\mathcal{R}$ to $\mathcal{R}_u$ in polynomial time. 
\end{lemma}
\begin{proof}
    The reduction simply outputs the solution on $\mathcal{R}_u$. To see that this is a valid solution, simply note that $\mathcal{R}_u \subseteq \mathcal{R}$.
\end{proof}

\begin{lemma}
    \label{lem: essentially_unique_ueopl}
    If $\mathcal{R} \in \cc{TF\Sigma_i^P}$ has essentially unique solutions and is $\mu$-d.s.r, then $\mathcal{R}_u \in \cc{TF\Sigma_i^P}$ is $\mu$-d.s.r with access to a \cc{\Sigma_{i-1}^P} oracle.
\end{lemma}

\znote{I modified the following proof}

\begin{proof}
    Let $V_1, V_2$ be the verifiers for $\mathcal{R}$. Let $\mathcal{A}^{\mathcal{R}, \cc{\Sigma_{i-1}^P}}$ be the $\mu$-d.s.r for $\mathcal{R}_u$. We will use this to construct $\mathcal{B}^{R_u, \cc{\Sigma_{i-1}^P}}$, a $\mu$-d.s.r for $\mathcal{R}$. $\mathcal{B}^{\mathcal{R}_u, \cc{\Sigma_{i-1}^P}}$ works as follows.
    
    On input $x$, $\mathcal{B}^{\mathcal{R}_u, \cc{\Sigma_{i-1}^P}}$ simulates $\mathcal{A}^{\mathcal{R}, \cc{\Sigma_{i-1}^P}}$ and obtains a solution $y$ such that $(x, y) \in \mathcal{R}$. This is feasible since an $\mathcal{R}_u$ oracle is also an $\mathcal{R}$ oracle by \cref{lem:R-to-Ru}. 

    If $V_1(x, y) = 0$, $\mathcal{B}^{\mathcal{R}_u, \cc{\Sigma_{i-1}^P}}$ returns $y$. Otherwise, it uses its $\cc{\Sigma_{i-1}^P}$ oracle to find the lexicographically smallest $y'$ such that $V_1(x, y') = 1$, and returns $y'$. 

    Notice that the reduction proceeds in polynomial time since $V_1$ can be verified in $\cc{\Sigma_{i-2}^P}$, we can use a $\cc{\Sigma_{i-1}^P}$ oracle to find the lexicographically smallest $y'$ such that $V_1(x, y') = 1$ if such a $y'$ exists (which it must whenever this case is encountered in our simulation due to the definition of $\mathcal{R}_u$). Correctness of $\mathcal{B}^{\mathcal{R}_u, \cc{\Sigma_{i-1}^P}}$ is inherited from the correctness $\mathcal{A}^{\mathcal{R}, \cc{\Sigma_{i-1}^P}}$.
\end{proof}

    


\begin{theorem} 
\label{thm:R-UEOPL}
    If $\mathcal{R} \in \cc{TF\Sigma_i^P}$ has essentially unique solutions and is (randomized) $\mu$-d.s.r with access to a $\cc{\Sigma_{i-1}^P}$ oracle, then $\mathcal{R}$ has a (randomized) reduction to $\cc{UEOPL}^{\cc{\Sigma_{i-1}^P}}$.
\end{theorem}
\begin{proof}
    This follows from \cref{lem:R-to-Ru} and \cref{thm:mu-dsr}.
\end{proof}

\section{New Upper Bounds for Avoid and the Linear Ordering Principle}\label{sec:lop}

In this section we will use the $\mu$-d.s.r framework to show that both \prob{Avoid} and \prob{Linear Ordering Principle} are in $\cc{UEOPL}^{\cc{NP}}$ (\cref{thm:main_avoid_ueopl} and \cref{thm:main_lop_ueopl_np}). 

We start by recalling the definition of \prob{Avoid} that was introduced by \cite{kleinberg2021total, korten2022hardest}.

\begin{mdframed}
    \textbf{Problem: \prob{Avoid}} 
    \paragraph{Input:} A polynomial sized circuit $C:\bit^n  \rightarrow \bit^{n+1}$ 
    \paragraph{Output:} Find an element in $\bit^{n+1}$ not in the range of the circuit $C$.
\end{mdframed}

It is clear that \prob{Avoid} describes a total problem since a solution is always guaranteed via the dual pigeonhole principle. Moreover any solution $y$ to \prob{Avoid} can be verified by an $\NP$ oracle since checking if $y$ is a valid solution corresponds to checking that $\forall x$: $C(x) \neq y$ which is just a $\coNP$ statement. As a result we can safely put \prob{Avoid} in $\cc{TF\Sigma_2^P}$.

To get an upper bound on \prob{Avoid} we follow the framework in \cite{korten2024LOP} by proving an upper bound on an adjacent problem \prob{Linear Ordering Principle} (\prob{LOP}) which Avoid can be reduced to. Moreover since \prob{LOP} has a unique solution, we end up being able to isolate a canonical solution for any instance of \prob{Avoid}. 

We now define the problem \prob{Linear Ordering Principle} introduced in \cite{korten2024LOP}

\begin{mdframed}
    \textbf{Problem: Linear Ordering Principle (LOP)} 
    \paragraph{Input:} A polynomial sized circuit $C:\bit^n  \times \bit^n \rightarrow \bit$, encoding a total order $\prec\;$ 
    \paragraph{Output:} 
        \begin{enumerate}
            \item Find a witness that $\prec\;$ does not define a total order i.e. $x, y, z \in \{0, 1\}^n$ such that one of the following holds: (i) $x \prec\; x$, (ii) $x \neq y$, $x \not \prec\; y$ and $y \not \prec\; x$ or (iii) $x \prec\; y \prec\; z$ and $x \not \prec\; z$. 
        \item Find the minimal element of the total order defined by $\prec\;$. 
        \end{enumerate}
\end{mdframed}

By definition, \prob{LOP} is a total problem. Verifying a witness of the violation of a total order can be done in polynomial time. To verify that a solution $y$ is actually the minimal element, we can use a $\NP$ oracle since it corresponds to checking the following $\coNP$ predicate $\forall x:$ $\prec\;(y, x) = 1$. As a consequence we have that \prob{LOP} is in $\cc{TF\Sigma_2^P}$, and  $\prob{LOP}$ has essentially unique solutions (\cref{def: essentially_unique}).



We will now show \prob{LOP} is $\mu$-d.s.r for polynomially bounded $\mu$. The main idea is to reduce the problem of finding a minimal element in the total order to finding minimal elements $a_0, a_1$ in the first and second halves of the total order respectively, and then comparing $a_0$ and $a_1$. This idea bears a resemblance to the proof in \cite{harsha2022downward} that the \cc{PLS}-complete problem \prob{Iter-with-Source} is d.s.r.
\begin{theorem} \label{lop:dsr}
    $\prob{LOP}$ is $\mu$-d.s.r with access to a $\cc{NP}$ oracle for polynomially bounded $\mu$. 
\end{theorem}
\begin{proof} 
    Let $\prec: \{0, 1\}^n \times \{0,1\}^n \rightarrow \{0 , 1\}$ be our $\prob{LOP}$ instance. We define $\mu(\prec) = n$. Clearly, $\mu(\prec) \leq |\prec|. $ If $\prec$ is not a total order, we can find the lexicographically smallest witness violating the total order using an $\cc{NP}$ oracle. We can ensure that the three types of violations have some lexicographic preference over each other. 

    Otherwise, $\prec$ must be a total order. For $i \in \{0, 1\}$, define $\prec_i\;:\{0, 1\}^{n-1} \times \{0, 1\}^{n-1} \rightarrow \{0, 1\}$ such that $\prec_i(x,y) = \prec(i ||x, i||y)$. We query the oracle on $\prec_0$ and $\prec_1$, each has a unique solution $a_0$ and $a_1$ respectively. We output $\min(a_0, a_1)$ with respect to $\prec$. 

    $\prec_0$ and $\prec_1$ are total since $\prec$ is total. Therefore, $a_i \prec x$ for all $x \in i || \{0, 1\}^{n-1}$. Therefore $\min(a_0, a_1) \prec x$ for all $x \in \{0, 1\}^n$, as required.  
\end{proof}

\begin{corollary}
    \label{cor: LOP_in_UEOPL}
    $\prob{LOP} \in \cc{UEOPL}^{\cc{NP}}, \prob{Avoid} \in \cc{UEOPL}^{\cc{NP}} $, and $\cc{L_2^P}$ reduces to $\cc{UEOPL}^{\cc{NP}}$. 
\end{corollary}
\begin{proof}
    The fact that $\prob{LOP} \in \cc{UEOPL}^{\cc{NP}}$ follows from \cref{lop:dsr} and \cref{thm:R-UEOPL}. This combined with the fact that \prob{Avoid} reduces to \prob{LOP} \cite{korten2024LOP} lets us conclude $\prob{Avoid} \in \cc{UEOPL}^{\cc{NP}}$. Finally, $\cc{L_2^P}$ is simply the class of all problems which are polynomial time many-one reducible to \prob{LOP} \cite{korten2024LOP}. 
\end{proof}

Finally, as a consequence of the connections between solving \prob{Avoid} and explicit construction of combinatorial objects we get the following corollary (\cref{cor:main_explicit_construction}):

\begin{corollary}
    There is an explicit construction of Ramsey Graphs, Two Source Extractors, Rigid Matrices, Linear Codes, Hard Truth Tables for any Fixed Polynomial Sized Circuit, $\cc{K}^{\poly}$-random strings in $\cc{UEOPL}^{\NP}$
\end{corollary}

\section{The Complexity of \prob{King}}\label{sec:king}

In this section we use ideas from the $\mu$-d.s.r framework to show that the problem \prob{King} lies in ($\cc{PLS}^{\cc{\Sigma_2^P}} \cap \cc{ZPP^{\Sigma_2^P}}$). In \cref{lem: king_choices} we identify a certain weak d.s.r property of \prob{King} that can be exploited, to design both $\cc{PLS}^{\cc{\Sigma_2^P}}$ and $\cc{ZPP}^\cc{\Sigma_2^P}$ algorithms. Surprisingly, the two approaches seem to use downward self-reducibility in very different ways.

Unlike the other total function problems we deal with, we will not be able to apply the $\mu$-d.s.r framework directly since we will find that $\mu$ is not polynomially bounded (and we therefore cannot apply \cref{thm:mu-dsr}). But the ideas from \cref{thm:mu-dsr} will carry over to \prob{King} naturally and allow us to show \prob{King} is in \cc{PLS^{\Sigma_2^P}} and design non-trivial algorithms for \prob{King}.

We begin by reviewing the definitions from \cite{kleinberg2021total} on what it means for a player in a tournament to be a king.
\begin{definition}
    A vertex $v$ in a digraph $G$ is called a king if every vertex in $G$ can be reached
    from $v$ by a path of length at most 2.
\end{definition}

\begin{definition}
     A digraph $G$ is called a tournament if for every pair of distinct vertices
    $u, v \in G$, exactly one of the directed edges $(u, v)$ or $(v, u)$ is present in $G$.
\end{definition}

In fact, there is a very simple lemma due to \cite{landau1953dominance} which shows:
\begin{lemma} \label{lem:king_totality}
    Every tournament has a king.
\end{lemma}

\begin{proof}
    Let $G$ be a digraph representing a tournament, and $\mathcal{N}(u)$ represent all vertices who have an edge incident from $u$ in $G$. We now argue that any node $v$ with maximum out-degree in $G$ is a king. Towards contradiction, suppose that $v$ is not a king. Then there must exist a vertex $u$, such that $u$ is not incident to $v$, and $u$ is not incident to any vertex in $\mathcal{N}(v)$. So the out-degree of $u \geq |\mathcal{N}(v)| + 1$, which is greater than out-degree of $v$, a contradiction.
\end{proof}

\cite{kleinberg2021total} define the search problem \prob{King} which sits in $\Sigma_3^P$ and whose totality follows from \cref{lem:king_totality}. Consider as input a circuit $C : [2^n] \times [2^n] \rightarrow \bit$ where one can imagine $C$ implicitly encoding the digraph $G$ corresponding to a tournament, i.e. $C(x,y) = 1$ implies that there is a directed edge from $x$ to $y$ in $G$. If $C$ is not a tournament then there exists $x \neq y$ such that $C(x,y) = C(y,x)$ which can easily be checked by an $\NP$ oracle. If $C$ is a tournament we are guaranteed to have a king.

\begin{mdframed}
    \textbf{Problem: KING}
    \paragraph{Input:} A circuit $C:[2^n] \times [2^n] \rightarrow \bit$ encoding a digraph
    \paragraph{Output:}
    
    \begin{enumerate}
        \item Find a distinct $x_1, x_2 \in [2^n]$ such that $C(x_1, x_2) = C(x_2, x_1)$. [A no witness, showing $C$ does not encode a tournament]
    
        \item Find an element $k \in [2^n]$ such that for every $x \in [2^n]\setminus \{ k\}$, either $C(k, x) = 1$, or there exists an element $j \in [2^n]\setminus \{ k, x\}$, such that $C(k, j) = 1$ and $C(j, x) = 1$. [A yes witness, a solution to \prob{king}].
    \end{enumerate}
    
    
\end{mdframed}

We note that deciding whether there is a king for a digraph can be written as a $\cc{\Sigma_3^P}$ predicate: $\exists k$,  $\forall x$,  $\exists j$ \textnormal{s.t.} $[C(k,x) = 1] \vee [C(k,j) = 1 \wedge C(j, x) = 1]$. By definition \prob{King} is total, and it is in $\cc{TF\Sigma_3^P}$ since it's solutions can be verified by a $\cc{\Sigma_2^P}$ verifier:
    \begin{enumerate}
        \item We can verify the no witness $(x_1, x_2)$ using an $\NP$ oracle. There exists $(x_1, x_2) \in [2^n]$ such that $C(x_1, x_2) = C(x_2, x_1)$.
        
        \item We can verify the yes witness $(k)$ using a $\cc{\Sigma_2^P}$ oracle. For all $x \neq k \in [2^n]$ , there exists $j$ such that either $C(k,x) = 1$ or $(C(k,j) = 1 \wedge C(j,x) = 1)$.
    \end{enumerate}




If we want to extract a solution for \prob{king} via the existence argument in \cref{lem:king_totality}, it appears that we need to solve some kind of generic counting problem (outside of the polynomial hierarchy) to compute the neighborhood of a given node. Interestingly, we show that we need much less. Our starting point is 
the following crucial lemma which could be thought of as a weak d.s.r property
about the structure of tournaments. 

\znote{modified \cref{lem: king_choices}} \snote{Looks good}

We start by defining what we call a \emph{weak king} of a subset of vertices.

\begin{definition}[Weak King]
    Let $G$ be a digraph on a set of vertices $V$ and let $U \subseteq V$ be a subset of vertices. A vertex $v \in U$ is called a weak king of $U$ if every vertex in $U$ can be reached from $v$ by a path of length at most 2 in $G$.
\end{definition}

We note that a weak king $v$ for $U \subseteq V$ may not be a king in the induced subgraph $G[U]$ since it is allowed to use edges from $G \setminus G[U]$. On the other hand, if $U = V$, then a weak king is also a king.

\begin{lemma}
    \label{lem: king_choices}
    Let $G$ be a tournament on a set of vertices $V$ where $|V| \geq 1$ and let $U \subseteq V$ be any subset of vertices. For any $v \in V \setminus U$, if $u$ is a weak king for $U$, then either $u$ or $v$ is a weak king for $U \cup \{v\}$.
\end{lemma}
\begin{proof}
    Consider any such $G, U \subseteq V$ and $v \in V \setminus U$. Let $u$ be a weak king for $U$. There can be two cases:
    \begin{enumerate}
        \item Suppose there exists a path of length at most 2 from $u$ to $v$ in $G$. Since $u$ is a weak king for $U$, this implies $u$ is also a king for $U \cup \{v\}$. 

        \item Suppose there is no path of length at most 2 from $u$ to $v$ in $G$. We will show that for any $x \in U$, there is a path of length at most 2 from $v$ to $x$. We can have three cases: 
        \begin{itemize}
            \item $x = u$: We know that $(u, v) \notin G$. Hence $(v, u) \in G$.
            \item $(u,x) \in G$: We know that $(x, v) \notin G$ for otherwise we found a path of length 2 from $u$ to $v$. Hence $(v, x) \in G$.
            \item $\exists\; w \in V$, such that $(u,w), (w,x) \in G$: We know that $(w, v) \notin G$ for otherwise there is a path of length 2 from $u$ to $v$. Hence $(v, w) \in G$ and we have the path $(v,w),(w,x)$.
        \end{itemize}
        Hence $v$ is a weak king for $U \cup \{v\}$.
    \end{enumerate}
\end{proof}

\cref{lem: king_choices} now suggests a natural d.s.r strategy. Given a tournament $G$ on a set of vertices $V$, finding a king is equivalent to finding a weak king for the whole vertex set $V$. Moreoever, we can reduce the problem of finding a weak king for vertices $[1, t]$ to that of finding a weak king for vertices $[1, t-1]$, get back a weak king $k'$, and then output either $k'$ or $t$, whichever is a weak king for $[1,t]$. Notice that by \cref{lem: king_choices}, either $k'$ or $t$ will be a weak king for $[1,t]$. This d.s.r does not show \cc{PLS^{\Sigma_2^P}} membership of \prob{King} because for a circuit $C$, $\mu(C)$ is the number of vertices in the tournament implicitly defined by $C$. $\mu$ is therefore not polynomially bounded. Fortunately, our proposed d.s.r is ``memoryless'' in that we do not need to keep a stack trace. Therefore, we do not need all the machinery of \cref{thm:mu-dsr}. We will use that insight to turn our proposed d.s.r into a reduction from \prob{King} to $\prob{Sink-of-DAG}^{\cc{\Sigma_2^P}}$.

\begin{theorem}
    \label{thm: king_pls}
    \prob{King} is in $\cc{PLS^{\Sigma_2^P}}$.
\end{theorem}

\znote{modified due to \cref{lem: king_choices}}

\begin{proof}
    As previously discussed, a $\cc{\Sigma_2^P}$ oracle is sufficient to check the answer to a \prob{King} instance efficiently. We reduce \prob{King} to \prob{Sink-of-DAG} with a \cc{\Sigma_2^P} oracle as follows. We are given a \prob{King} instance $C: [2^n] \times [2^n] \rightarrow \{ 0, 1\}$. We first use our \cc{\Sigma_2^P} oracle to check if there exists distinct $x_1, x_2 \in [2^n]$ such that $C(x_1, x_2) = C(x_2, x_1)$. If there are, we can use the \cc{NP} oracle to find $x_1, x_2$ and output them as a type 1 solution to \prob{King}. Notice that in this case the reduction runs in $\poly(n)$ time and is correct. We will therefore assume for the remainder of the proof that $C$ defines a proper tournament $G$.
    
    We now show the reduction from \prob{King} to \prob{Sink-Of-Dag}. We first specify the successor circuit $S: [2^n] \times [2^n] \rightarrow [2^n] \times [2^n]$. We should think of the first input $i$ as an integer and the second input $x$ as a vertex. 
    $S$ is defined as follows:
    \begin{enumerate}
        \item 
        If $x$ is not a weak king for $[1,i]$ (which one can check using a \cc{\Sigma_2^P} oracle), $S(i, x) = (i, x)$.
        \item 
        For all $x$, $S(2^n, x) = (2^n, x)$.
        \item 
        If $x$ is a weak king for $[1, i+1]$ (which one can check using a \cc{\Sigma_2^P} oracle), $S(i, x) = (i+1, x)$.
        \item 
        If $i+1$ is a weak king for $[1, i+1]$ (which one can check using a \cc{\Sigma_2^P} oracle), $S(i, x) = (i+1, i+1)$.
    \end{enumerate}

    Notice that $S$ covers all possible cases by \cref{lem: king_choices}.

    We now specify the value circuit $V: [2^n] \times [2^n] \rightarrow [2^{n}]$, $V(i, k) = i$. The reduction calls the \prob{Sink-of-DAG} oracle on $S, V$ and receives back $(j, k)$ as output. Let $(j', k') = S(j, k)$. The reduction outputs $k'$. It should be clear that the reduction runs in time $\poly(n)$. We now show correctness. By the definition of \prob{Sink-of-Dag}, $S(j, k) \neq (j, k)$ but one of the following two conditions hold.
    \begin{enumerate}
        \item $S(S(j, k)) = S(j, k)$: In other words $S(j', k') = (j', k')$. Notice that by the definition of $S$, since $S(j, k) \neq (j, k)$, $k$ must be a weak king for $[1,j]$. Again by the definition of $S$, $k'$ is a weak king for $[1, j' = j + 1]$. Since $S(j', k') = (j', k')$, but $k'$ is a weak king for $[1, j' = j + 1]$, it must be the case that $j' = 2^n$ (otherwise $S(j', k')$ would equal $(j'+1, k'')$ for some $k''$). Therefore, $k'$ is a weak king for $[1, 2^n]$, and hence a king for $G$ as desired. 
        \item $V(S(j, k)) \leq V(j, k)$: Notice that since $S(j, k) \neq (j, k)$, by the definition of $S$, either $S(j, k) = (j+1, k)$ or $S(j, k) = (j+1, j+1)$. In either case $V(S(j, k)) = j+1$. So it cannot be the case that $V(S(j, k)) \leq V(j, k)$.
    \end{enumerate}
\end{proof}

To our surprise, we also find that \prob{King} is in \cc{TFZPP^{\Sigma_2^P}}. This proof uses different ideas from downward self-reducibility, but interestingly still relies on \cref{lem: king_choices}. 
\begin{theorem}
    \label{thm: king_tfzpp}
    \prob{King} is in \cc{TFZPP^{\Sigma_2^P}}
\end{theorem}

We will need a uniform sampler for $\cc{\Sigma_2^P}$ relations, which follows from observing that the sampler for $\NP$ relations in \cite{BGP00NPsampler} directly applies for higher classes in the polynomial hierarchy.

\begin{lemma}[\cite{BGP00NPsampler}]\label{thm: sigma_i_sampler}
    For $i \geq 1$, let $R$ be a \cc{\Sigma_i^P} relation. Then there is a uniform generator for $R$ which is implementable in probabilistic, polynomial time with a \cc{\Sigma_i^P} oracle.
\end{lemma}

\begin{proof}[Proof of \cref{thm: king_tfzpp}]
    We are given a \prob{King} instance $C: [2^n] \times [2^n] \rightarrow \{ 0, 1\}$. We first use our \cc{\Sigma_2^P} oracle to check if there exists distinct $x_1, x_2 \in [2^n]$ such that $C(x_1, x_2) = C(x_2, x_1)$. If there are, we can use the \cc{NP} oracle to find $x_1, x_2$ and output them as a type 1 solution to \prob{King}. Notice that in this case the reduction runs in $\poly(n)$ time and is correct. We will therefore assume for the remainder of the proof that $C$ defines a proper tournament $G$.
    

    For any vertex $u\in V$, we use $W_u$ to denote the set of all vertices    witnessing that $u$ is not a king. Formally,
    \[ W_u:=\{v \in V \setminus \{u\}: (u,v) \notin G, \forall w\in V \setminus \{u, v\}, (u,w)\notin G \lor (w,v)\notin G\} \; . \]

    In particular, if $u$ is a king, then $W_u = \emptyset$. For any given $u \in V$, we could sample uniformly from $W_u$ in probabilistic polynomial time with a \cc{\Sigma_2^P} oracle by \cref{thm: sigma_i_sampler}.
    
    Next, we describe the algorithm for finding a king. Start with a vertex $u \in V$, if $u$ is a king, we find a solution and terminate. Otherwise we sample $v$ uniformly from $W_u$ and repeat the process with $v$.

    To see that the algorithm above terminates in (expected) polynomial-time, we have the following two claims:
    \begin{claim} \label{clm: strict_decrease}
        For any $u \in V$ and $v \in W_u$, $W_v \subsetneq W_u$.
    \end{claim}
    
\znote{modified due to \cref{lem: king_choices}}
    
    \begin{proof}
        Let $U = V \setminus W_u$. By definition of $W_u$, $u$ is a weak king for $U$. 
        
        Now by \cref{lem: king_choices}, $v$ is a weak king of $U \cup \{v \}$ since $v \in W_u$. 
        
        It follows that $W_v \subseteq V \setminus (U \cup \{v\}) = W_u \setminus \{v\} \subsetneq W_u$.
    \end{proof}
    \begin{claim}\label{clm:const_frac}
        For any $u \in V$,
        \[ \Pr_{v \sim W_u} \left[|W_v| \leq \frac{2|W_u|}{3}\right] \geq \frac{1}{4} \; . \]
    \end{claim}
    \begin{proof}
        Let $s = |W_u|$. Consider the induced tournament $G[W_u]$ on the vertex set $W_u$. The total outdegree is exactly $s(s-1)/2$. By a simple counting argument, at least $1/4$ of the vertices in $W_u$ have outdegree at least $(s-1)/3$. Moreover, any vertex $w$ that is an outgoing neighbor of $v$ cannot be in $W_v$. Combining these two facts yields
        \[\Pr_{v \sim W_u}\left[|W_v| \leq (s - 1) - \frac{s-1}{3}\right] \geq \frac{1}{4} \; .\]
    \end{proof}
    Now we can measure the progress of our algorithm by $|W_u|$, which in expectation decreases by a constant factor in constant number of iterations. It terminates when $|W_u| = 0$ which takes in expectation $O(\polylog(|V|) = O(\poly(n))$ iterations.
\end{proof}

\cref{thm: king_tfzpp} stands in striking contrast to the following theorem. This indicates that the \cc{\Sigma_2^P} oracle is in some sense making \prob{King} much easier.
\begin{theorem}[\cite{mande2023randomized}]
    \label{thm: mande}
    In the black-box query model, any randomized algorithm for \prob{King} requires at least $\Omega(2^n)$ time.
\end{theorem}

Finally, as a side note, we show that there exists a faster than trivial algorithm to solve \prob{King} if we are given access to an \cc{NP} oracle.
\begin{theorem}
    \label{thm: king_efficient}
    There exists a $\poly(n) 2^n$ time algorithm using an \cc{NP} oracle algorithm for \prob{King}.
\end{theorem}

    

\znote{modified due to \cref{lem: king_choices}}

\begin{proof}
    The algorithm $\mathcal{A}$ is as follows. $\mathcal{A}$ uses the $\NP$ oracle to first find distinct $x, y \in [2^n]$ such that $C(x, y) = C(y, x)$ (if any) and if it finds them outputs $x, y$ as a type 1 solution to \prob{King}. This can be achieved in polynomial time with access to an $\NP$ oracle.  We therefore assume that $C(x, y) \neq C(y, x)$ for all $x \neq y$ for the remainder of the proof. 
    
    The algorithm runs for $2^n$ steps. At step $i$, we aim to find a weak king $v_i$ for $[1,i]$. Clearly, at step $1$, $v_1 = 1$ is a weak king of $[1,1]$. When we are at step $i$, \cref{lem: king_choices} tells us that a weak king of $[1,i]$ is either $v_{i-1}$ or $i$. Furthermore, it provides an easy check to determine if $v_{i-1}$ or $i$ is a weak king. If there is a length $2$ path in from $v_{i-1}$ to $i$, then $v_i = v_{i-1}$ is a weak king for $[1,i]$, otherwise $v_i = i$ is a weak king for $[1,i]$. We can check this using the $\NP$ oracle which tells us if there exists a $z \in V$ such that $C(v_{i-1}, z) = 1$ and $C(z, i) = 1$. Finally, the algorithm outputs $v_{2^n}$, the weak king of $[1,2^n]$ which is by definition a king.

    Since each step of the algorithm takes at most $\poly(n)$ time with access to the $\NP$ oracle, we also get our desired runtime. Correctness follows from \cref{lem: king_choices}.
\end{proof}

Our \cc{NP} oracle algorithm for \prob{King} compares favorably to lower bounds for finding a king in a tournament without access to an oracle.
\begin{theorem}[\cite{shen2003searching}]
    \label{thm: shen}
    In the black-box query model, any deterministic algorithm for \prob{King} requires at least $\Omega(2^{4n/3})$ time.
\end{theorem}

\cref{thm: shen} and \cref{thm: mande} were originally stated for the problem of finding a king in a tournament on $t$ vertices. They showed a $\Omega(t^{4/3})$ query lower bound for deterministic algorithms and a $\Omega(2^n)$ query lower bound for randomized algorithms respectively. Setting $t = 2^n$ yields \cref{thm: shen} and \cref{thm: mande}. We find it somewhat interesting that access to a \cc{NP} provides a provable speedup for deterministic algorithms (in the query model) from $\Omega(2^{4n/3})$ to $\poly(n) 2^n$ and that access to a $\cc{\Sigma_2^P}$ oracle provides a provable speedup for randomized algorithms (in the query model) from $\Omega(2^n)$ to $\poly(n)$.

\section{Downward self-reducibility in \cc{TFNP}}\label{sec: PLS_TFNP}

We now demonstrate the utility of the $\mu$-d.s.r framework by applying it to several \cc{TFNP} problems. For \prob{P-LCP} and \prob{Memdet}, we are able to dramatically simplify proofs of \cc{UEOPL} and \cc{PLS} membership respectively. Interestingly, unlike most proofs of membership in a \cc{TFNP} subclass, our proofs of \cc{PLS}/\cc{UEOPL} membership do not mirror (or even resemble) the proof of totality for the base problem. For the \prob{Tarski} problem, we give a proof of \cc{PLS} membership which mirrors the divide-and-conquer algorithm for finding Tarski fixed points. Although this proof of \cc{PLS} membership is not syntactically simpler than the original \cite{etessami2020tarski}, we view it as conceptually simpler since it only requires knowledge of the well-known divide-and-conquer algorithm for \prob{Tarski}. Finally, we show \cc{UEOPL} membership for the \prob{S-Arrival} problem. This proof is admittedly significantly more complicated than the original proof \cite{gartner2018arrival}\footnote{technically this work only showed \cc{CLS} membership, but it was observed in \cite{fearnley2020unique} that the proof also implies \cc{UEOPL} membership}. However, we have chosen to include it because we hope it will provide an alternative perspective which may be useful in future analysis of the \prob{S-Arrival} problem.

\subsection{P-LCP}
\label{subsubsec: p-lcp}

We begin by defining a P-matrix, the linear complementarity problem, and the \prob{Promise-P-LCP} problem.
\begin{definition}
    A matrix $M \in \mathbb{R}^{n \times n}$ is a P-matrix if every principal minor is positive.
\end{definition}

\begin{definition}
    For any $M \in \mathbb{R}^{n \times n}$ and $q \in \mathbb{R}^{n}$, we say $z \in \mathbb{R}^{n}$ is a solution to the linear complementarity problem (LCP) if all the following conditions hold.
    \begin{enumerate}
        \item $z \geq 0$,
        \item $y = q+Mz \geq 0$,
        \item $z^T y = 0$.
    \end{enumerate}
\end{definition}


\begin{mdframed}
    \textbf{Problem: \prob{Promise-P-LCP}} 
    \label{def: promise-plcp}
\paragraph{Input:} A P-matrix $M \in \mathbb{R}^{n \times n}$ and $q \in \mathbb{R}^{n}$. 
\paragraph{Output:} A solution $z$ to the linear complementarity problem with inputs $M$ and $q$.
\end{mdframed}

The following lemma will be useful in showing a $\mu$-d.s.r for \prob{Promise-P-LCP} and immediately implies uniqueness of solutions for \prob{Promise-P-LCP}.
\begin{lemma}[\cite{samelson1958partition}]
    \label{lem: P-matrix}
     $M \in \mathbb{R}^{n \times n}$ is a P-matrix if and only if for any $q \in \mathbb{R}^n$, LCP with input $M,q$ has exactly one solution.
\end{lemma}

\begin{lemma}
    \label{lem: p-lcp-dsr}
    \prob{Promise-P-LCP} $\mu$-d.s.r for polynomially bounded $\mu$.
\end{lemma}
\begin{proof}
    For an input matrix $M \in \mathbb{R}^{n \times n}$ and $q \in \mathbb{R}^n$, we define the measure $\mu(M) = n$. We now show a downward self-reduction for \prob{Promise-P-LCP}. If $-M^{-1}q \geq 0$, output $-M^{-1}q$. Otherwise, let $M_i \in \mathbb{R}^{(n-1) \times (n-1)}$ denote $M$ with row $i$ and column $i$ removed. Similarly, let $q_i \in \mathbb{R}^{n-1}$ denote $q$ with entry $i$ removed. The reduction queries its oracle on $(M_1, q_1), \dots, (M_n, q_n)$ to get back the answers $z_1, \dots, z_n \in \mathbb{R}^{n-1}$ respectively. We find $i$ such that $z_i$ is a solution to $\prob{LCP}(M, q)$ after inserting a $0$ in the $i$-th coordinate, and output the solution.

    This reduction clearly runs in $\poly(n)$ time. It is also a downward reduction since $\mu(M_j) = n-1$ for all $j \in [1,n]$. We see that it is also promise-preserving since if $M$ is a P-matrix, then $M_j$ is also a P-matrix for all $j \in [1, n]$.

    If the reduction outputs $z = M^{-1}q$, then clearly $z \geq 0$, $y = q + Mz = 0$, and $z^T y = 0$. Otherwise, consider $z^*$, the solution to $\prob{LCP}$ on input $M, q$ (and note that $z^*$ is unique by \cref{lem: P-matrix}). If $z^*$ contains no zero coordinates, then $y = (0, \dots, 0)$ since $z^{*T}y = 0$. But this implies $z^* = -M^{-1}q$, which cannot be true since that case has already been checked. So let $i$ be any coordinate of $z^*$ that is zero and consider $z^*_{-i} \in \mathbb{R}^{n-1}$, $z^*$ with the $i^{\text{th}}$ coordinate removed. $z^*_{-i}$ must be a solution to the \prob{LCP} on input $M_i, q_i$ since $z^*$ was a solution to \prob{LCP} on input $M, q$. Notice that by \cref{lem: P-matrix}, $z^*_{-i}$ is the only solution to $(M_i, q_i)$. Therefore, by the correctness of the d.s.r, $z_i = z^*_{-i}$, and the reduction outputs $z^*$ as desired.
\end{proof}

\begin{corollary}\label{cor:plcp_ueopl}
    \prob{Promise-P-LCP} reduces to \cc{UEPOL}.
\end{corollary}
\begin{proof}
    \prob{Promise-P-LCP} is clearly a \cc{Promise\Sigma_1^P} problem. \cref{lem: P-matrix} shows that it has unique solutions and \cref{lem: p-lcp-dsr} shows that it is $\mu$-d.s.r for polynomially bounded $\mu$. We can therefore apply \cref{thm:mu-dsr} to show that \prob{Promise-P-LCP} reduces to \cc{UEPOL}.
\end{proof}

\subsection{Graph Games}
\label{subsubsec: graph-games}

We begin by defining graph games \cite{beckmann2008complexity}.
\begin{definition}[Graph Games, \cite{beckmann2008complexity}]
\label{def: graph_games}
$G = (V_0, V_1, E)$ \text{ is a graph game of size } n \text{ if}
\begin{enumerate}
    \item $V_i$ are the positions of player $P_i$, for $i = 0, 1$. They have to satisfy $V_0 \cap V_1 = \emptyset$, and $V_0 \cup V_1 \subseteq \{1, \ldots, n\}$. $V := V_0 \cup V_1$ is the set of all positions.
    \item $E \subseteq V \times V$ is the set of possible moves.
    \item In graph-theoretic terms, $V$ is the set of nodes, and $E$ the set of edges of graph $G$. They have to satisfy in addition that at least one edge is leaving each node.
\end{enumerate}
If $v \in V_i$, we say that player $P_i$ owns $v$.
\end{definition}

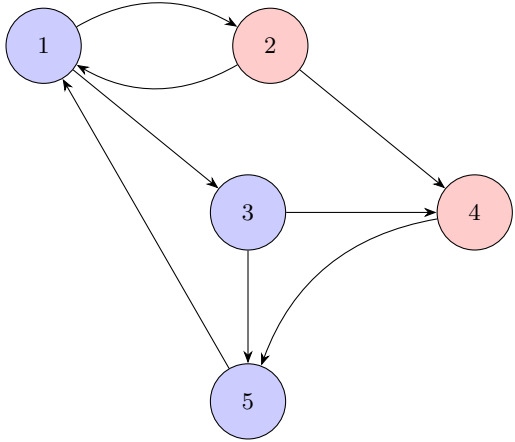
\begin{figure}
    \centering
\usetikzlibrary{arrows.meta, positioning, calc}
\begin{tikzpicture}[
    node distance = 1.5cm and 2cm,
    game node/.style = {circle, draw, minimum size=1cm, font=\footnotesize},
    even node/.style = {game node, fill=blue!20},
    odd node/.style = {game node, fill=red!20},
    >=Stealth
]

\node[even node] (1) {1};
\node[odd node] (2) [right=of 1] {2};
\node[even node] (3) [below right=of 1] {3};
\node[odd node] (4) [below right=of 2] {4};
\node[even node] (5) [below=of 3] {5};

\draw[->] (1) to[bend left] (2);
\draw[->] (2) to[bend left] (1);
\draw[->] (1) -- (3);
\draw[->] (2) -- (4);
\draw[->] (3) -- (4);
\draw[->] (3) -- (5);
\draw[->] (4) to[bend right] (5);
\draw[->] (5) -- (1);

\end{tikzpicture}
    \caption{An example graph game}
    \label{fig: graph-game}
\end{figure}

As an example, in \cref{fig: graph-game}, $V_0 = \{ 1, 3, 5\}, V_1 = \{ 2, 4\},$ and $ E = \{ (1,2), (2,1), (1,3), (5, 1), (2,4), (3,4), (3, 5), (4, 5)\}$.

\begin{definition}[Playing and Winning, \cite{beckmann2008complexity}]
\label{def: graph_game_play}
A play from a node $v \in V$ is an infinite path $v = v_1 \rightarrow v_2 \rightarrow v_3 \rightarrow \ldots$ in $G$ with each edge $v_i \rightarrow v_{i+1} \in E$ chosen by the player owning $v_i$. The winner of a play is the player owning the least node which is visited infinitely often in the play.
\end{definition}

As an example, in \cref{fig: graph-game}, a play may be $1 \rightarrow 2 \rightarrow 1 \rightarrow 3 \rightarrow 4 \rightarrow 5 \rightarrow \cdots$. In this example, the player 1 (the blue player), made moves $1, 3, 4$ and player 2 (the red player), made moves $2$ and $5$.

\begin{definition}
    Let $G = (V_0, V_1, E)$ and $V = V_0 \cup V_1$. We say $G$ has a memoryless deterministic strategy for player $i$ if there exists a function $\sigma: V_i \rightarrow V$ such that for all $v \in V_i$, $(v, \sigma(v)) \in E$, and $\sigma$ is a dominant strategy for player $i$. In particular, even if the other player uses a non-memoryless strategy, player $i$ still wins by playing $\sigma$.
\end{definition}

\begin{lemma}[\cite{beckmann2008complexity}]
    \label{lem: graph_games_NP}
    For every simple graph game $G$, a memoryless deterministic strategy exists, and can be verified in polynomial time.
\end{lemma}

We are now ready to define the \cc{TFNP} problem corresponding to simple graph games, and thus to parity games.

\begin{mdframed}
    \textbf{Problem: \prob{Memdet}} 
    \label{def: memdet}

\paragraph{Input:} A canonical representation of a graph game $G = (V_0, V_1, E)$
\paragraph{Output:} A winner $i \in \{0, 1\}$ and a memoryless deterministic strategy $\sigma: V_i \rightarrow V$ for player $i$.
\end{mdframed}

\cref{lem: graph_games_NP} shows that \prob{Memdet} is in \cc{TFNP}. Now, we will use downward self-reducibility to show that \prob{Memdet} is in \cc{PLS}. 

\begin{theorem}
    \prob{Memdet} is $\mu$-downward self-reducible for a polynomially bounded $\mu$. 
\end{theorem}
\begin{proof}
    For a \prob{Memdet} instance $G = (V_0, V_1, E)$, we define $\mu(G) = |E| - |V_0| - |V_1|$. Notice that when $\mu(G) = 0$, each vertex has exactly one edge coming out of it. Therefore, each player has exactly one possible memoryless deterministic strategy. We can therefore determine which strategy wins in polynomial time. In this case, \prob{Memdet} is solvable in polynomial time. It should also be clear that $\mu(G) \leq |G|$.

    We now show \prob{Memdet} is $\mu$-d.s.r. Assume $\mu(G) > 0$. We say that player $i$ has a forced strategy if each vertex in $V_i$ has exactly one edge coming out of it. The reduction first checks if there is a forced strategy which also a winning strategy for either player. If so, it outputs the solution it finds. Otherwise, let $G_i$ denote $G$ with edge $i$ removed. Notice however that $G_i$ is not necessarily a graph game as defined in \cref{def: graph_games} since it may be the case that $G_i$ does not have at least one edge leaving each vertex. So define $\mathcal{G} = \{ G_i : i \in E, \text{$G_i$ is a graph game}\}$. The self reduction constructs $\mathcal{G}$ and calls the \prob{Memdet} oracle on every graph in $\mathcal{G}$ to get back solutions of the form $(p, \sigma_i)$ where $\sigma_i$ is a memoryless deterministic strategy for player $p$ to win in $G_i$. The reduction then checks for every $(p, \sigma_i)$ it receives from the oracle if $\sigma_i$ is a memoryless, deterministic strategy for $G$, and outputs the first $(p, \sigma_i)$ that is. If no such strategy is found (which we will show never happens), the reduction outputs $0$.

    We first show that the reduction is valid. The construction of $\mathcal{G}$ clearly runs in polynomial time. Checking whether $\sigma_i$ is a winning strategy for $G$ can be done in polynomial time and therefore so is checking all at most $|G|$ of them. Therefore, the reduction runs in polynomial time. Furthermore, we see that $\mu(G_i) = \mu(G)-1$ for every $i \in E$, so the reduction is indeed a \emph{downward} self-reduction.

     We now show correctness. We only show the reduction succeeds when $\mu(G) \geq 1$, since otherwise, the \prob{Memdet} is solvable in polynomial time. We must show that one of the $(p, \sigma_i)$ is a solution to \prob{Memdet} on $G$. Let $\sigma^*$ be a memoryless, deterministic strategy for $G$. Notice that since $\mu(G) \geq 1$, there must exist an edge $i^*$ we can eliminate to create $G_{i^*}$ such that $(p, \sigma^*)$ is a solution to \prob{Memdet} on $G_{i^*}$. Here we consider two cases. 
    \begin{enumerate}
        \item 
        Say $p$ has a forced strategy in for $G$. The reduction clearly succeeds.
        \item 
        Say player $p$ does not have a forced strategy for $G$. Let us consider what happens when $i^*$ is any edge belonging to $p$ which is not chosen in $\sigma^*$. Then player $p$ wins in $G_{i^*}$ since $\sigma^*$ is clearly a winning strategy for $G_{i^*}$. Therefore, the solution to $G_{i^*}$ is $(p, \sigma')$ for some $\sigma'$. Notice that $(p, \sigma')$ is a solution for $G$ since the player opposing $p$ has the same options in $G_{i^*}$ as they do in $G$. Therefore, any solution to $G_{i^*}$ is a solution to $G$.
    \end{enumerate}
\end{proof}

\begin{corollary}\label{cor:memdet}
    \prob{Memdet} is in \cc{PLS}.
\end{corollary}

We note that the problem of finding memoryless deterministic strategies for parity games reduces to \prob{Memdet} \cite{beckmann2008complexity}. Therefore, the problem of finding memoryless deterministic strategies for parity games is also in \cc{PLS}.

\subsection{Tarski} \label{sec:tarski}

Consider any integers $L, d \geq 1$. For an integer lattice $[L]^d$, we define a partial order $\preceq$ such that $x \preceq y$ if for all $i \in [d]$, $x_i \leq y_i$. We now define the \cc{TFNP} problem of finding a Tarski fixed point \cite{tarski1955lattice, etessami2020tarski}.

\begin{mdframed}
    \textbf{Problem: \prob{Tarski}} 
    \label{def: tarski}

\paragraph{Input:} $N = 2^n, n, d \geq 1$. Given $C: [N]^d \rightarrow [N]^d$ with $\mathrm{size}(C) \leq \poly(d, n)$. 
\paragraph{Output:}
\begin{enumerate}
    \item $x \in [N]^d$ such that $C(x) = x$. 
    \item  distinct $x, y \in [N]^d$ such that $x \preceq y$ but $C(x) \npreceq C(y)$. 
\end{enumerate} 
\end{mdframed}

Although it is unclear how this problem is itself downward self-reducible, we can define a new more general problem, which we call \prob{Tarski+}, which is indeed downward self-reducible. For $\ell, h \in [N]^d$, let $L(\ell, h) = \{ x \in [N]^d : \ell \preceq x \preceq h \}$.


\begin{mdframed}
    \textbf{Problem: \prob{Tarski+}} 
    \label{def: tarski+}

\paragraph{Input:} Let $N = 2^n, n, d \geq 1, t \in [2, d+1]$. Given $C: [N]^d \rightarrow [N]^d$ with $\mathrm{size}(C) \leq \poly(d, n)$ \\
$m_t, \dots, m_d \in [N]$, $\ell, h \in [N]^{t-1}$. \\ 
Define $f:L(\ell, h) \rightarrow [N]^{t-1}$ as $f(m_1, \dots, m_{t-1}) = C(m_1, \dots, m_{t-1}, m_t, \dots, m_d)_{[1, t-1]}$.
\paragraph{Output:}
\begin{enumerate}
    \item $x \in [N]^d$ such that $C(x) = x$. 
    \item  distinct $x, y \in [N]^d$ such that $x \preceq y$ but $C(x) \npreceq C(y)$. 
    \item $x$ such that $f(x) \notin L(\ell, h)$.
\end{enumerate} 
\end{mdframed}

In the above definition, we imagine that $m_i$ is some $n$-bit string encoding $\bot$ if $i \in [t, d]$. \prob{Tarski+} simply lets consider \prob{Tarski} problem on a sublattice $L(\ell, h)$ where $\ell, h \in [N]^{t-1}$ and fixed $m_t, \dots, m_d$. Our downward self-reduction for \prob{Tarski+} is now simply based on the classic divide-and-conquer algorithm for finding Tarski fixed points. See \cite{dang2011computational} for a description of this algorithm, or \cite{etessami2020tarski} for a compressed description.

\begin{lemma}
    \label{lem: tarski+_dsr}
    \prob{Tarski+} is $\mu$-d.s.r for polynomially bounded $\mu$.
\end{lemma}
\znote{Updated $\mu$, please double check.} \snote{Looks good.}
\begin{proof}
    Consider a \prob{Tarski+} instance $T = (C, m_t, \dots, m_d, \ell, h)$. 
    Let $\mu(T) = t + \lceil \log_2(|L(\ell, h)|)\rceil$. Note that $\mu$ is clearly polynomially bounded in the instance size. We now show a $\mu$-d.s.r for $T$. By definition, $f: L(\ell, h) \rightarrow [N]^{t-1}$ is such that $f(m_1, \dots, m_{t-1}) = C(m_1, \dots, m_{t-1}, m_t, \dots, m_d)_{[1, t-1]}$. If $t=2$ or $|L(\ell, h)| \leq 100$, the reduction just solves the instance (in the case $t=2$ via a binary search, and in the case $|L(\ell, h)| \leq 100$ by brute force).

    Otherwise, let $t' = t - 1$ and $m = m_{t'} = \lceil (\ell_{t'} + h_{t'})/2 \rceil$. The reduction calls its oracle on $C, m_{t'}, m_t, \dots, m_d, \ell_{[1, t'-1]}, h_{[1, t'-1]}$ and receives an answer back. Depending on the answer, we can have a number of cases. First, we will show how to find a solution for the current instance in each case and argue its correctness. 

    First, let's handle the simpler case where the oracle returns a solution of type 2 or 3: 

    \begin{enumerate}
        \item If it receives a type 2 solution $x, y \in [N]^{t'-1}$, it outputs a type 2 solution $(x, m), (y, m)$. Let $f'(x) = f(x, m)_{[1, t'-1]}$ be the implicit function in the oracle call. Clearly if $x \preceq y$ but $f'(x) \npreceq f'(y)$, we also have that $(x, m) \preceq (y, m)$ but $f(x, m) = (f'(x) , a) \npreceq (f'(y), b) = f(y, m)$, for some $a, b \in [N]$. This is because $\npreceq$ is preserved under concatenating any element on either side. 
        \item If it receives a type 3 solution $x \in [N]^{t'-1}$ such that $f'(x) \notin L(\ell_{[1, t'-1]}, h_{[1, t'-1]})$. The reduction outputs a type 3 solution $(x, m)$. Note $f(x, m)_{[1, t'-1]} \notin L(\ell_{[1, t'-1]}, h_{[1, t'-1]})$. Notice that this property is persevered by adding coordinates, so $f(x, m) \notin L(\ell, h)$, and $(x, m)$ forms a type 3 solution to our instance. 
    \end{enumerate}
    
    Otherwise, the reduction receives a type 1 solution $x^* \in [N]^{t'-1}$ such that $f'(x^*) = x^*$. Here the output $x^*$ of the first oracle call has the property $f(x^*, m)_{[1, t']} = x^*$ and for all $i \in [1, t'-1]$, $\ell_i \leq x^*_i \leq h_i$. We can have four cases: 
    \begin{enumerate}
        \item
        If $f(x^*, m)_{t'} \notin [\ell_{t'}, h_{t'}]$, then output $(x^*, m)$ as a type 3 solution.
        Note $f(x^*, m)_{t'} \notin [\ell_{t'}, h_{t'}]$, then $f(x^*, m) \notin L(\ell, h)$. But clearly $(x^*, m) \in L(\ell, h)$ since for all $i \in [1, t'-1]$, $\ell_i \leq x^*_i \leq h_i$ and $\ell_{t'} \leq m \leq h_{t'}$. Therefore $(x^*, m)$ is a type 3 solution.
        \item
        If $f(x^*, m)_{t'} = m$, the reduction outputs $(x^*, m)$ as a type 1 solution since $f(x^*, m) = (x^*, m)$. 
        
        \item 
        If $f(x^*, m)_{t'} > m$, the the reduction calls its oracle on the input $C, m_t, \dots, m_d, f(x^*, m), h$. If the oracle returns a type 1 solution, $z$, the reduction outputs $z$. If the reduction receives a type 2 solution, $z_1, z_2$, the reduction outputs $z_1, z_2$.  These are both valid since $L(f(x^*, m), h) \subseteq L(\ell, h)$. 
        
        If the reduction receives a type 3 solution $z \in [N]^{t'}$, this implies $z \in L(f(x^*, m), h)$ but $f(z) \notin L(f(x^*, m), h)$. We can have two cases: 
        
        \begin{itemize}
            \item If $f(x^*, m) \npreceq f(z)$, output $(x^*, m), z$ as a type 2 solution. Note that $(x^*, m) \preceq f(x^*, m) \preceq z$ where the second inequality follows from the fact that $z \in L(f(x^*, m), h)$ and the first inequality follows from the fact that $f'(x^*) = x^*$ and $f(x^*, m)_{t'} > m$. Therefore, $(x^*, m) \preceq z$ but $f(x^*, m) \npreceq f(z)$, so $(x^*, m), z$ forms a type 2 solution to our instance.
            \item Otherwise $f(z) \npreceq h$, then output $z$ as a type 3 solution. Since $f(z) \notin L(\ell, h)$, $z$ is a type 3 solution to the instance.
        \end{itemize}
        
        \item 
        If $f(x^*, m)_{t'} < m$, the reduction calls its oracle on the input $C, m_t, \dots, m_d, \ell, f(x^*, m)$. If the oracle returns a type 1 solution, $z$, the reduction outputs $z$. If the reduction receives a type 2 solution, $z_1, z_2$, the reduction outputs $z_1, z_2$. If the reduction receives at type 3 solution $z$ and $f(z) \npreceq f(x^*, m)$, output $z, f(x^*, m)$ as a type 2 solution. Otherwise, output $z$ as a type 3 solution. Correctness follows by a similar argument as in the previous case. 
    \end{enumerate}

    We have proved the correctness of our reduction. Finally, we show that the potential $\mu$ decreases by at least one at each step. Clearly in the first oracle call in the $\mu$-d.s.r, the number of points in $L(\ell, h)$ does not go up, but $t$ goes down by at least 1, so $\mu$ goes down by at least one. In the other oracle call, $t$ stays the same, but by our choice of $m$, $|L(\ell, h)|$ goes down by at least a factor 2. In particular, the lattice the oracle is queried on is a sub-lattice of $L(\ell, h)$ where the last coordinate is only allowed to range over half the values it was originally allowed to range over. Therefore, $\mu$ decreases by at least 1.
\end{proof}

\begin{corollary}
    \label{cor: tarski}
    \prob{Tarski} is in \cc{PLS}.
\end{corollary}
\begin{proof}
    \prob{Tarski} reduces to \prob{Tarski+} by stetting $t = d+1$ and $\ell =(0, \dots, 0), h=(N, \dots, N)$. \prob{Tarski+} is in \cc{PLS} by \cref{lem: tarski+_dsr} and \cref{cor: mu-d.s.r}. Therefore, \prob{Tarski} is in \cc{PLS}.
\end{proof}

\subsection{\prob{S-Arrival}}

\prob{Arrival} is a zero-player variant of graph games (\cref{def: graph_games}). Here, instead of players making the decision of where to move, we imagine a train which is making the decision all by itself. We imagine a graph on $n$ vertices where each vertex $i$ has exactly two directed edges coming out of it $e_i^0, e_i^1$. The $j^{\text{th}}$ time the train arrives at vertex $i$, it takes the edge $e_i^{(\text{$j \mod 2$})}$. The \prob{Arrival} problem asks us to prove that the train will or will not eventually make it to its destination $d$. We now define the \prob{Arrival} problem more formally.

\begin{definition}[Switch Graph, \cite{gartner2018arrival}]
A switch graph is a tuple $G = (V, E, s_0, s_1)$ where $s_0, s_1 : V \to V$ and $E = \{(v, s_i(v)) \mid \forall v \in V, i \in \{0, 1\} \}$.
\end{definition}

\begin{algorithm}
\label{alg: run}
\caption{\prob{Run} \cite{gartner2018arrival}}
\begin{algorithmic}
\Require A switch graph $G = (V, E, s_0, s_1)$ and two vertices $o, d \in V$.
\Ensure For each edge $e \in E$, the number of times the train traversed $e$.
\State $v \gets o$ \Comment{position of the train}
\State $\forall u \in V$, set $s_{curr}[u] \gets s_0(u)$ and $s_{next}[u] \gets s_1(u)$
\State $\forall e \in E$, set $r[e] \gets 0$ \Comment{initialize the run profile}
\State $step \gets 0$
\While{$v \neq d$}
\State $w \gets s_{curr}[v]$ \Comment{compute the next vertex}
\State $r[s_{curr}[v]]++$ \Comment{update the run profile}
\State swap$(s_{curr}[v], s_{next}[v])$
\State $v \gets w$ \Comment{move the train}
\State $step \gets step + 1$
\EndWhile
\State \Return $r$
\end{algorithmic}
\end{algorithm}

\begin{definition}[\prob{Arrival}, \cite{dohrau2017arrival}]
Given a switch graph $G = (V, E, s_0, s_1)$ and two vertices $o, d \in V$, the \prob{Arrival} problem is to decide whether the algorithm \prob{Run} (\cref{alg: run}) terminates, i.e., whether the train reaches the destination $d$ starting from the origin $o$.
\end{definition}

For an instance $G$ of \prob{Arrival}, we first define some notation. A run is simply an ordered tuple of edges $(e_1, \dots, e_t)$. When starting at $o$ on a graph $G$, if the train traverses $e_1, \dots, e_t$ in order for some consecutive portion, then $(e_1, \dots, e_t)$ is called a run for $G, o$. The run profile of a run $m = (e_1, \dots, e_t)$ is a vector $v$ indexed by the edges in $E$ such that $v_{e}$ is the number of times $e$ appears in $m$. We say that a run profile $r$ for $G, o$ is valid if there exists a run $m = (e_1, \dots, e_t)$ such that $m$ is a run for $G$ and $r$ is the run profile of $m$. Given two runs $m$ and $m'$ such that $(m, m')$ is also a valid run, it is not hard to see that $r_m + r_{m'}$ is a valid run profile for $(m, m')$. We will often use this correspondence between concatenating two runs and easily being able to compute their combined run profile from their individual run profiles. Finally, we refer to the set of edges specified by the map $s_{\text{curr}}$ as the active edges.

The above problem is a decision problem, but to define the search problem \prob{S-Arrival}, we need the validity of run profiles to be efficiently verifiable. For a destination $d$, let $V_{good}(d)$ denote the set of all vertices $v$ for which there exists a directed path in $(V, E)$ from $v$ to $d$ and let $V_{bad}(d)$ denote $V \setminus V_{good}$.

\begin{lemma}[\cite{dohrau2017arrival}]
    \label{lem: arrival_tfnp}
    For an arrival instance $G = (V, E, s_0, s_1)$, $o, d$, the train eventually reaches $d$ from $o$ or it eventually reaches a vertex in $V_{bad}(d)$ from $o$. Furthermore, it always reaches one of these nodes in at most $n 2^n$ steps.
\end{lemma}

Notice that if the train reaches $d$ from $o$, the run ends. If the train reached $V_{\text{bad}}(d)$ from $o$, then it has no hope of reaching $d$, since there is no path possible from any vertex in $V_{\text{bad}}(d)$ to $d$. Since there are atmost $n2^n$ steps in either case, the entries in the run profile sum up to $n2^n$. This means the run profile can be represented by $O(n)$ bits. 

\begin{lemma}[\cite{gartner2018arrival}]
    \label{lem: arrival_NP}
    For any run profile of length $O(n)$ bits, there exists an $\poly(n)$ time algorithm to check if the run profile is valid. Furthermore, for a valid run profile, the final vertex of the run can be determined in polynomial time.
\end{lemma}

\cref{lem: arrival_tfnp} and \cref{lem: arrival_NP} together tell us that there exists a witness that train reaches $d$ or not. If the train reaches $d$, the run profile that ends at $d$ is a valid and efficiently checkable witness of this fact. If the train does not reach $d$, the run profile that ends at $V_{\text{bad}}$ (without going through an outgoing edge of $d$) is an efficiently checkable witness of this fact.

We are now ready to define the \prob{S-Arrival} problem. We use the original definition from \cite{karthik2017arrival}:

\begin{mdframed}
    \textbf{Problem: \prob{S-Arrival}} 
    \label{def: ramsey}

\paragraph{Input:} A switch graph $G = (V, E, s_0, s_1)$ and two vertices $o, d \in V$.
First, we define a graph $G'$ as follows: 

\begin{itemize}
    \item Add a new vertex $\bar{d}$. 
    \item For each vertex $v \in V_{bad}$, set $s_0(v) = \bar{d}$ and $s_1(v) = \bar{d}$. 
    \item Edges $s_0(d), s_1(d), s_0(\bar{d}), s_1(\bar{d})$ are self-loops. 
\end{itemize}
\paragraph{Output:}
The task is to find a run profile $r$ starting at $o$ and ending at either $d$ or $\bar{d}$ (such that either is visited exactly once)
\end{mdframed}

We now confirm for ourselves that \prob{S-Arrival} is in \cc{TFUP}. In polynomial time, we can find the vertices in $V_{bad}$. This will allow us to construct the graph $G'$. By \cref{lem: arrival_tfnp}, we know that the run profile can be represented by at most $O(n)$ bits. By \cref{lem: arrival_NP}, we can verify if the run profile is valid in polynomial time. From the run profile, we can also verify that the final vertex in the run is either $d$ or $\bar{d}$ and that they only appear once. This proves that \prob{S-Arrival} is in \cc{TFUP}. 
We now show that \prob{S-Arrival} is d.s.r.
\begin{theorem}
    \label{thm: arrival_dsr}
    \prob{S-Arrival} is $\mu$-d.s.r. for some polynomially bounded $\mu$. 
\end{theorem}

\begin{proof}[Proof Sketch]
    
    We define our measure $\mu(G) = |V|$, the number of vertices in $G$. In this proof, we will use run profiles interchangeably with any runs they might represent. The only modification to runs we will do is concatenating two runs together, thereby also ensuring that the new run profile can also be easily computed.

    We will downward self-reduce an \prob{S-Arrival} instance $I$ where $\mu(I) = n$. $I$ consists of $G = (V, E, s_0, s_1)$ and $o, d \in V$. Let $V_{good}$ be the set of vertices which have an outgoing edge to $d$ and $E_{good}$ be the corresponding edges. If $V_{good}$ was empty, we could output the run $(o, \bar{d})$. Now, consider any $v \in V_{good}$ such that $(v, d) \in E_{good}$. Modify the graph $G$ to produce $G^{(1)} = (V^{(1)} = V \;\backslash\; \{d\}, E^{(1)}, s_0^{(1)}, s_1^{(1)})$ such that each edge that leads to $d$ ($(u, d) \in E$) leads to $v$ instead ($(u, v) \in E^{(1)}$). First, observe that $V_{bad} = V_{bad}^{(1)}$. 
    
    Let $m^{(1)}$ be a run that is a solution to \prob{S-Arrival} on $G^{(1)}, o^{(1)} = o, d^{(1)} = v$. Now, we can have a few cases: 
    \begin{enumerate}
        \item $m^{(1)}$ ends at node $\bar{d}$: we return $m^{(1)}$ since it is also a valid run in $G'$ that ends at $\bar{d}$ since no such run can pass through a node in $V_{good}$.
        \item $m^{(1)}$ ends at node $v$ via an edge $(u, d) \in E_{good}$: we can easily modify the run to include $(u, d)$ instead of $(u, v)$ to produce a valid run $m$ in $G'$ that ends at $d$. 
        \item $m^{(1)}$ ends at node $v$ via an edge $(u, v) \not \in E_{good}$ and $(v, d) \in s_0$: we can modify $m^{(1)}$ by adding an edge $(v, d)$ to produce a valid run $m$ in $G'$ ending at $d$. 
        \item If $m^{(1)}$ ends at node $v$ via an edge $(u, d) \not \in E_{good}$ and $(v, d) \in s_1$: In this case, since we have already reached $v$ once, we must reach $v$ again in order to use $(v, d)$ to reach $d$. Let $o^{(2)} = s_0(v)$. Note $(v, o^{(2)})$ is the next edge visited by a run in $G'$. Modify $m^{(1)}$ to concatenate the edge $(v, s_0(v))$. Next, we modify $G^{(1)}$ to reflect the fact that we have already traversed the run $m^{(1)}$. For every edge $e \in s_i^{(1)}$, $e \in s_i^{(2)}$ iff $e$ is visited an even number of times in the run $m^{(1)}$. Note that this can be computed by the run profile of $m^{(1)}$. Now, we use the oracle to receive a solution $m^{(2)}$ to \prob{S-Arrival} on $G^{(2)}, o^{(2)}, v$. Now, we can have three cases reflecting the previous three cases: 
        \begin{enumerate}
            \item $m^{(2)}$ ends at node $\bar{d}$: we return $m^{(1)}$ concatenated with $m^{(2)}$ since it is also a valid run in $G'$ that ends at $\bar{d}$.
            \item $m^{(2)}$ ends at node $v$ via an edge $(u, d) \in E_{good}$: we can easily modify the run, changing $(u, v)$ to $(u, d)$ instead. Now, $m^{(1)}$ concatenated with $m^{(2)}$ is a valid run in $G'$ ending at $d$. 
            \item $m^{(2)}$ ends at node $v$ via an edge $(u, d) \not \in E_{good}$: We return the run $m$ which is $m^{(1)}$ concatenated with $m^{(2)}$ and $(v, d)$. We know that the run $m^{(1)}$ concatenated with $m^{(2)}$ ends at $v$ and visits $v$ exactly twice. Furthermore, $(v, d) \in s_1$, therefore the next edge taken by the train will be $(v, d)$ which makes $m$ a valid run in $G'$ ending in $d$. 
        \end{enumerate}
        
    \end{enumerate}

    


    We have showed the correctness of our reduction. This reduction clearly runs in polynomial time as each graph modification and run (profile) modification can be carried out in polynomial time as noted above. Finally, the reduction has the downward property. It only calls the oracle on smaller instances as $\mu(G^{(1)}) = \mu(G^{(2)}) = \mu(G) - 1$.

\end{proof}

\begin{corollary}\label{cor:s_arrival}
    \prob{S-Arrival} is in \cc{UEOPL}.
\end{corollary}
\begin{proof}
    The fact that \prob{S-Arrival} is in \cc{TFUP} and d.s.r (\cref{thm: arrival_dsr}) implies that it is in \cc{UEOPL} by \cref{cor: mu-d.s.r}.
\end{proof}

\section{Limitations of self-reducibility as a framework in \cc{TFNP}}
\subsection{Not all \cc{PLS}-complete problems are traditionally d.s.r}
\cite{harsha2022downward} asked if (traditional) downward self-reducibility is a complete framework for \cc{PLS}-complete problems. Recall that traditional downward self-reducibility means that the oracle can only answer queries of input length strictly less than $n$, rather than smaller in some size function $\mu$. This would be a rather surprising characterization of \cc{PLS}-completeness. Here, we give a negative answer to this question under the assumption that $\cc{PLS} \neq \cc{FP}$. To do so, we will exhibit a problem which is \cc{PLS}-complete but not traditionally d.s.r.

\begin{mdframed}
    \textbf{Problem: \prob{RestrictedSizeSoD}} 
    \label{def: restricted-size iter}

\paragraph{Input:} A bitstring $x$ of length $2^{100 \cdot 2^{k}}$ for some integer $k$.

\paragraph{Output:}
If $|x| \neq 2^{100 \cdot 2^{k}}$ for some integer $k$, $0$ is the only valid answer. Otherwise, let $x = 0^t \circ 1 \circ x'$ for some integer $t$ and string $x'$. $y$ is a solution if and only if $y$ is a solution to the \prob{Sink-of-DAG} on input $x'$.
\end{mdframed}

In \prob{RestrictedSizeSoD}, one should think of only certain input sizes being valid. In particular, input sizes of length $2^{2^k}$ for some integer $k$ (where we allow padding) are valid. To ensure totality though, in cases when the input is not of length $2^{2^k}$ for some integer $k$, we say exactly the trivial solution $0$ is a solution.

\begin{lemma}
    \label{lem: res-in-pls}
    \prob{RestrictedSizeSoD} is in \cc{PLS}.
\end{lemma}
\begin{proof}
    Given an input $x \in \{ 0, 1\}^{n}$ to \prob{RestrictedSizeSOD}, we can solve it using a \cc{PLS} oracle as follows. If $n \neq 2^{2^k}$ for any integer $k$, we output $0$. Otherwise, we let $x = 0^t \circ 1 \circ x'$ for some integer $t$ and string $x'$, use our \cc{PLS} oracle to solve $\prob{Sink-of-DAG}$ on $x'$, and output that solution. The fact that the reduction runs in polynomial time and is correct is immediate.
\end{proof}

\begin{lemma}
    \label{lem: res-pls-hard}
    \prob{RestrictedSizeSoD} is \cc{PLS}-hard.
\end{lemma}
\begin{proof}
    Say $x' \in \{ 0, 1\}^n$ is some input to \prob{Sink-of-DAG}. Let $k = \lceil \log_2(\log_2 n) \rceil + 1$ and $n' = 2^{2^k}$. We can reduce \prob{Sink-of-DAG} on $x'$ to \prob{RestrictedSizeSoD} as follows. Let $x = 0^{n'-n-1} \circ 1 \circ x'$. We now feed $x$ to our \prob{RestrictedSizeSoD} oracle to get a solution $y$, which we then output.

    The reduction clearly runs in polynomial time since $n' \leq 2^{2^{\log_2(\log_2 n)+2}} = 2^{4 \cdot \log_2(n)} = n^4$. Correctness follows from the definition of \prob{RestrictedSizeSoD}.
\end{proof}

\begin{lemma}
    \label{lem: poly-pls}
    If \prob{RestrictedSizeSoD} is traditionally d.s.r, then \prob{RestrictedSizeSoD} has a polynomial time algorithm.
\end{lemma}
\begin{proof}
    Say that \prob{RestrictedSizeSoD} has a traditional d.s.r. We assume without loss of generality that this d.s.r only queries its oracle on inputs of size $2^{2^k}$ for some integer $k$ since otherwise, we could simulate the answers to those queries as being $0$ ourselves. Furthermore, we assume that the number of queries made by the d.s.r on an input of size $2^{2^k}$ is exactly $(2^{2^k})^c$. Furthermore, all other operations in the d.s.r (including the ones that move the tape-head and prepare it for the d.s.r call) require time exactly $(2^{2^k})^{c'}$ for some constant $c$.

    Let us now consider running the recursive algorithm defined by the recursive algorithm for \prob{RestrictedSizeSoD}. Let $T(k)$ denote the time the recursive algorithm takes on an input of size $2^{2^k}$. Notice that $T(0) = O(1)$.
    \[ T(k) = (2^{2^k})^c \cdot T(k-1) + (2^{2^k})^{c'}\]

    Since $T(k-1) \geq 1$ and $c' < c$, we can say the following.
    \[ T(k) = 2 \cdot (2^{2^k})^c \cdot T(k-1)\]

    Let $d = \max(2c, c')+1$. We will show by induction that $T(k) \leq (2^{2^k})^d$. Notice that this clearly holds for $k = 0$.

    Now for the inductive case.
    \begin{align*}
        T(k) &= (2^{2^k})^c \cdot T(k-1) + (2^{2^k})^{c'}\\
        &\leq (2^{2^k})^c \cdot (2^{2^{k-1}})^d + (2^{2^k})^{c'}\\
        &= (2^{2^k})^c \cdot (2^{2^{k}})^{\frac{d}{2}} + (2^{2^k})^{c'}\\
        &= (2^{2^k})^{c + \frac{d}{2}}+ (2^{2^k})^{c'}
    \end{align*}

    Note that by our choice of $d$, $c+d/2 \leq (d-1)/2 + d/2 = d-1/2$ and $c' \leq d-1$. Therefore, the following holds.
    \begin{align*}
        T(k) &\leq (2^{2^k})^{c + \frac{d}{2}}+ (2^{2^k})^{c'}\\
        &\leq (2^{2^k})^{d-\frac{1}{2}}+ (2^{2^k})^{d-1}\\
        &= (2^{2^k})^{d-1} \left[ 2^{2^{k-1}} + 1\right]\\
        &\leq (2^{2^k})^{d-1} \left[ 2^{2^{k}} \right]\\
        &= (2^{2^k})^d
    \end{align*}
    The second to last inequality holds since $k \geq 1$.
    Therefore, the runtime of our recursive algorithm for \prob{RestrictedSizeSoD} on an input of length $2^{2^k}$ is at most $(2^{2^k})^d$ for some constant $d$. Therefore, our algorithm's running time is polynomial in its input length, as desired.
\end{proof}

\begin{theorem}
    If every \cc{PLS}-complete problem is traditionally d.s.r, then $\cc{PLS} = \cc{FP}$.
\end{theorem}
\begin{proof}
    If every \cc{PLS}-complete problem is traditionally d.s.r, then \prob{RestrictedSizeSoD} is traditionally d.s.r (since \prob{RestrictedSizeSoD} is \cc{PLS}-complete by \cref{lem: res-in-pls} and \cref{lem: res-pls-hard}). This implies \prob{RestrictedSizeSoD} can be solved in polynomial time by \cref{lem: poly-pls}. However, since \prob{RestrictedSizeSoD} is \cc{PLS}-complete, this implies $\cc{PLS} = \cc{FP}$.
\end{proof}

We have therefore shown that it is unlikely that every \cc{PLS}-complete problem is traditionally d.s.r. We leave open the question of whether every \cc{PLS}-complete problem is $\mu$-d.s.r.
\begin{observation}
    \label{obs: np_not_dsr}
    The exact same technique of defining a padded problem shows that not every \cc{NP}-complete problem is traditionally d.s.r unless $\cc{P} = \cc{NP}$ and not every \cc{PSPACE}-complete problem is traditionally d.s.r unless $\cc{P} = \cc{PSPACE}$.
\end{observation}

\subsection{A note on random self-reducibility}
\label{subsec: rsr}
We now turn our attention to a different question asked in \cite{harsha2022downward}. We say a problem is random self-reducible (r.s.r) if it has a worst-case to average-case reduction. The notion of average-case here is with respect to some efficiently samplable distribution. \cite{harsha2022downward} ask if every r.s.r \cc{TFNP} problem is in some syntactic \cc{TFNP} sublcass (e.g. \cc{PPP}). 

An explicit construction problem is one where the input is a string $1^n$ and we are asked to generate some combinatorial object of length $\poly(n)$ satisfying some property. A classic example is the problem of generating an $n$-bit prime larger than $2^n$ given as input $1^n$.

We simply note here that all \cc{TFNP} explicit construction problems can be reduced to an r.s.r problem. Consider a \cc{TFNP} explicit construction $A$ with verifier $V(1^n, \cdot)$. Consider the problem $B$ where as input we are given a string from $\{ 0, 1\}^n$ and are asked to output a string which $y$ such that $V(1^n, y) = 1$. Note that $A$ trivially reduces to $B$. Furthermore, $B$ is trivially uniformly self-reducible since only its input length matters.

Therefore, if all \cc{TFNP} r.s.r problems belonged to some \cc{TFNP} subclass $\cc{C}$, then all \cc{TFNP} explicit construction problems would belong to \cc{C}. This would be quite a surprising theorem and we view this as a weak barrier to a \cref{thm:mu-dsr} type theorem for random self-reducibility.

\section{Acknowledgements}

The authors would like to thank Alexander Golovnev, and Noah Stephens-Davidowitz
for many helpful discussions and feedback on an earlier draft of this manuscript. 
The authors would also like to thank the anonymous referees for useful comments.

\bibliographystyle{alpha}
\bibliography{refs}

\newcommand{\etalchar}[1]{$^{#1}$}
\begin{thebibliography}{KKMP21}

\bibitem[Ald83]{aldous1983minimization}
David Aldous.
\newblock Minimization algorithms and random walk on the $ d $-cube.
\newblock {\em The Annals of Probability}, 11(2):403--413, 1983.

\bibitem[All10]{allender2010new}
Eric Allender.
\newblock New surprises from self-reducibility.
\newblock In {\em Programs, Proofs, Processes, Conference on Computability in Europe}, pages 1--5. Citeseer, 2010.

\bibitem[BCH{\etalchar{+}}22]{bitansky2022ppad}
Nir Bitansky, Arka~Rai Choudhuri, Justin Holmgren, Chethan Kamath, Alex Lombardi, Omer Paneth, and Ron~D Rothblum.
\newblock Ppad is as hard as lwe and iterated squaring.
\newblock In {\em Theory of Cryptography Conference}, pages 593--622. Springer, 2022.

\bibitem[BGP00]{BGP00NPsampler}
Mihir Bellare, Oded Goldreich, and Erez Petrank.
\newblock Uniform generation of np-witnesses using an np-oracle.
\newblock {\em Inf. Comput.}, 163(2):510–526, December 2000.

\bibitem[BM08]{beckmann2008complexity}
Arnold Beckmann and Faron Moller.
\newblock On the complexity of parity games.
\newblock In {\em Visions of Computer Science-BCS International Academic Conference}. BCS Learning \& Development, 2008.

\bibitem[Cai07]{cai2007s2p}
Jin-Yi Cai.
\newblock \( \mathsf{S2P} \subseteq \mathsf{ZPP^{NP}} \).
\newblock {\em Journal of Computer and System Sciences}, 73(1):25--35, 2007.

\bibitem[Can96]{canetti1996more}
Ran Canetti.
\newblock More on bpp and the polynomial-time hierarchy.
\newblock {\em Information Processing Letters}, 57(5):237--241, 1996.

\bibitem[CHLR23]{chen2023range}
Yeyuan Chen, Yizhi Huang, Jiatu Li, and Hanlin Ren.
\newblock Range avoidance, remote point, and hard partial truth table via satisfying-pairs algorithms.
\newblock In {\em Proceedings of the 55th Annual ACM Symposium on Theory of Computing}, pages 1058--1066, 2023.

\bibitem[CHR24]{chen2024symmetric}
Lijie Chen, Shuichi Hirahara, and Hanlin Ren.
\newblock Symmetric exponential time requires near-maximum circuit size.
\newblock In {\em Proceedings of the 56th Annual ACM Symposium on Theory of Computing}, pages 1990--1999, 2024.

\bibitem[CL24]{chen2024hardness}
Yilei Chen and Jiatu Li.
\newblock Hardness of range avoidance and remote point for restricted circuits via cryptography.
\newblock In {\em Proceedings of the 56th Annual ACM Symposium on Theory of Computing}, pages 620--629, 2024.

\bibitem[DGK{\etalchar{+}}17]{dohrau2017arrival}
J{\'e}r{\^o}me Dohrau, Bernd G{\"a}rtner, Manuel Kohler, Ji{\v{r}}{\'\i} Matou{\v{s}}ek, and Emo Welzl.
\newblock {ARRIVAL: A zero-player graph game in NP $\cap$ coNP}.
\newblock {\em A Journey Through Discrete Mathematics: A Tribute to Ji{\v{r}}{\'\i} Matou{\v{s}}ek}, pages 367--374, 2017.

\bibitem[DQY11]{dang2011computational}
Chuangyin Dang, Qi~Qi, and Yinyu Ye.
\newblock Computational models and complexities of tarski’s fixed points.
\newblock Technical report, Technical report, 2011.

\bibitem[EPRY20]{etessami2020tarski}
Kousha Etessami, Christos Papadimitriou, Aviad Rubinstein, and Mihalis Yannakakis.
\newblock Tarski’s theorem, supermodular games, and the complexity of equilibria.
\newblock In {\em 11th Innovations in Theoretical Computer Science Conference (ITCS 2020)}, pages 18--1. Schloss Dagstuhl--Leibniz-Zentrum f{\"u}r Informatik, 2020.

\bibitem[FGMS20]{fearnley2020unique}
John Fearnley, Spencer Gordon, Ruta Mehta, and Rahul Savani.
\newblock Unique end of potential line.
\newblock {\em Journal of Computer and System Sciences}, 114:1--35, 2020.

\bibitem[GGNS23]{gajulapalli2023range}
Karthik Gajulapalli, Alexander Golovnev, Satyajeet Nagargoje, and Sidhant Saraogi.
\newblock Range avoidance for constant depth circuits: Hardness and algorithms.
\newblock {\em Approximation, Randomization, and Combinatorial Optimization. Algorithms and Techniques}, 2023.

\bibitem[GHH{\etalchar{+}}18]{gartner2018arrival}
Bernd G{\"a}rtner, Thomas~Dueholm Hansen, Pavel Hub{\'a}cek, Karel Kr{\'a}l, Hagar Mosaad, and Veronika Sl{\'\i}vov{\'a}.
\newblock Arrival: Next stop in cls.
\newblock In {\em 45th International Colloquium on Automata, Languages, and Programming (ICALP 2018)}. Schloss-Dagstuhl-Leibniz Zentrum f{\"u}r Informatik, 2018.

\bibitem[GLV24]{gajulapalli2024oblivious}
Karthik Gajulapalli, Zeyong Li, and Ilya Volkovich.
\newblock Oblivious complexity classes revisited: Lower bounds and hierarchies.
\newblock In {\em 44th IARCS Annual Conference on Foundations of Software Technology and Theoretical Computer Science}, page~25, 2024.

\bibitem[GLW22]{guruswami2022range}
Venkatesan Guruswami, Xin Lyu, and Xiuhan Wang.
\newblock Range avoidance for low-depth circuits and connections to pseudorandomness.
\newblock {\em ACM Transactions on Computation Theory}, 2022.

\bibitem[HMR22]{harsha2022downward}
Prahladh Harsha, Daniel Mitropolsky, and Alon Rosen.
\newblock Downward self-reducibility in tfnp.
\newblock {\em arXiv preprint arXiv:2209.10509}, 2022.

\bibitem[HV25]{hirsch2025upper}
Edward~A Hirsch and Ilya Volkovich.
\newblock Upper and lower bounds for the linear ordering principle.
\newblock {\em arXiv preprint arXiv:2503.19188}, 2025.

\bibitem[ILW23]{ilango2023indistinguishability}
Rahul Ilango, Jiatu Li, and R~Ryan Williams.
\newblock Indistinguishability obfuscation, range avoidance, and bounded arithmetic.
\newblock In {\em Proceedings of the 55th Annual ACM Symposium on Theory of Computing}, pages 1076--1089, 2023.

\bibitem[JW24]{jin2024faster}
Ce~Jin and Hongxun Wu.
\newblock A faster algorithm for pigeonhole equal sums.
\newblock {\em arXiv preprint arXiv:2403.19117}, 2024.

\bibitem[{Kar}17]{karthik2017arrival}
{Karthik C. S.}
\newblock Did the train reach its destination: The complexity of finding a witness.
\newblock {\em Information Processing Letters}, 121:17--21, 2017.

\bibitem[KKMP21]{kleinberg2021total}
Robert Kleinberg, Oliver Korten, Daniel Mitropolsky, and Christos Papadimitriou.
\newblock Total functions in the polynomial hierarchy.
\newblock In {\em 12th Innovations in Theoretical Computer Science Conference (ITCS 2021)}. Schloss Dagstuhl-Leibniz-Zentrum f{\"u}r Informatik, 2021.

\bibitem[Kor22a]{korten2022derandomization}
Oliver Korten.
\newblock Derandomization from time-space tradeoffs.
\newblock In {\em 37th Computational Complexity Conference (CCC 2022)}. Schloss-Dagstuhl-Leibniz Zentrum f{\"u}r Informatik, 2022.

\bibitem[Kor22b]{korten2022hardest}
Oliver Korten.
\newblock The hardest explicit construction.
\newblock In {\em 2021 IEEE 62nd Annual Symposium on Foundations of Computer Science (FOCS)}, pages 433--444. IEEE, 2022.

\bibitem[KP24]{korten2024LOP}
Oliver Korten and Toniann Pitassi.
\newblock Strong vs. weak range avoidance and the linear ordering principle.
\newblock In {\em 2024 IEEE 65th Annual Symposium on Foundations of Computer Science (FOCS)}, pages 1388--1407, 2024.

\bibitem[KVM99]{klivans1999graph}
Adam~R Klivans and Dieter Van~Melkebeek.
\newblock Graph nonisomorphism has subexponential size proofs unless the polynomial-time hierarchy collapses.
\newblock In {\em Proceedings of the Thirty-First Annual ACM Symposium on Theory of Computing}, pages 659--667, 1999.

\bibitem[Lan53]{landau1953dominance}
Hyman~Garshin Landau.
\newblock On dominance relations and the structure of animal societies: Iii the condition for a score structure.
\newblock {\em The bulletin of mathematical biophysics}, 15:143--148, 1953.

\bibitem[Li24]{li2024symmetric}
Zeyong Li.
\newblock Symmetric exponential time requires near-maximum circuit size: Simplified, truly uniform.
\newblock In {\em Proceedings of the 56th Annual ACM Symposium on Theory of Computing}, pages 2000--2007, 2024.

\bibitem[LY22]{li20221}
Jiatu Li and Tianqi Yang.
\newblock $3.1n-o(n)$ circuit lower bounds for explicit functions.
\newblock In {\em Proceedings of the 54th Annual ACM SIGACT Symposium on Theory of Computing}, pages 1180--1193, 2022.

\bibitem[MPS23]{mande2023randomized}
Nikhil~S Mande, Manaswi Paraashar, and Nitin Saurabh.
\newblock Randomized and quantum query complexities of finding a king in a tournament.
\newblock {\em arXiv preprint arXiv:2308.02472}, 2023.

\bibitem[PPY23]{pasarkarshortchoice2023}
Amol Pasarkar, Christos Papadimitriou, and Mihalis Yannakakis.
\newblock {Extremal Combinatorics, Iterated Pigeonhole Arguments and Generalizations of PPP}.
\newblock In Yael Tauman~Kalai, editor, {\em 14th Innovations in Theoretical Computer Science Conference (ITCS 2023)}, volume 251 of {\em Leibniz International Proceedings in Informatics (LIPIcs)}, pages 88:1--88:20, Dagstuhl, Germany, 2023. Schloss Dagstuhl -- Leibniz-Zentrum f{\"u}r Informatik.

\bibitem[RS98]{russell1998symmetric}
Alexander Russell and Ravi Sundaram.
\newblock Symmetric alternation captures bpp.
\newblock {\em computational complexity}, 7(2):152--162, 1998.

\bibitem[RSW22]{ren2022range}
Hanlin Ren, Rahul Santhanam, and Zhikun Wang.
\newblock On the range avoidance problem for circuits.
\newblock In {\em 2022 IEEE 63rd Annual Symposium on Foundations of Computer Science (FOCS)}, pages 640--650. IEEE, 2022.

\bibitem[SSW03]{shen2003searching}
Jian Shen, Li~Sheng, and Jie Wu.
\newblock Searching for sorted sequences of kings in tournaments.
\newblock {\em SIAM Journal on Computing}, 32(5):1201--1209, 2003.

\bibitem[STW58]{samelson1958partition}
Hans Samelson, Robert~M Thrall, and Oscar Wesler.
\newblock A partition theorem for euclidean n-space.
\newblock {\em Proceedings of the American Mathematical Society}, 9(5):805--807, 1958.

\bibitem[Tar55]{tarski1955lattice}
Alfred Tarski.
\newblock A lattice-theoretical fixpoint theorem and its applications.
\newblock 1955.

\bibitem[Zha24]{zhang2024pigeonhole}
Stan Zhang.
\newblock Pigeonhole equal subset sum in \({O}^*(2^{n/3}) \).
\newblock Master's thesis, Massachusetts Institute of Technology, 2024.

\end{thebibliography}

\end{document}